\documentclass{ieee-ojvt}
\usepackage{cite}
\usepackage{amsmath,amssymb,amsfonts}
\usepackage{algorithmic}
\usepackage{graphicx,color}
\usepackage{textcomp}

\usepackage{amsthm,bm}
\usepackage{xcolor}
\usepackage{booktabs}
\usepackage{tabularx}
\usepackage{array}        % For advanced table column control\emph{
\usepackage[caption=false, font=footnotesize]{subfig}
\newtheorem{theorem}{Theorem}
\newtheorem{lemma}{Lemma}
\newtheorem{Corollary}{Corollary}
\newtheorem{rem}{Remark}
\newtheorem{prop}{Proposition}

\def\BibTeX{{\rm B\kern-.05em{\sc i\kern-.025em b}\kern-.08em
    T\kern-.1667em\lower.7ex\hbox{E}\kern-.125emX}}
\AtBeginDocument{\definecolor{ojcolor}{cmyk}{0.93,0.59,0.15,0.02}}
\def\OJlogo{\vspace{-12pt}\includegraphics[height=24pt]{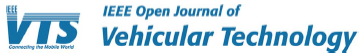}}
\begin{document}
%\receiveddate{XX Month, XXXX}
%\reviseddate{XX Month, XXXX}
%\accepteddate{XX Month, XXXX}
%\publisheddate{XX Month, XXXX}
%\currentdate{27 June, 2024}
%\doiinfo{OJVT.2024.0627000}

\title{The ISAC Tradeoff Cliff: Fundamental Limits under Waveform Uncertainty and Finite Blocklength}

\author{Mohammed Zafar Ali Khan\authorrefmark{1}, SENIOR MEMBER, IEEE,  AND Lajos Hanzo\authorrefmark{1}, LIFE FELLOW, IEEE
}
\affil{Indian Institute of  
Technology Hyderabad,Telangana 502284, India}
\affil{School of ECS, University of Southampton, 4004 UK}
\corresp{CORRESPONDING AUTHOR: Lajos Hanzo (e-mail: hanzo@soton.ac.uk).}
%\authornote{This work was supported ?.}
\markboth{The ISAC Tradeoff Cliff: Fundamental Limits under Waveform Uncertainty and Finite Blocklength}{Mohammed Zafar Ali Khan and Lajos Hanzo}

\begin{abstract}
Integrated Sensing and Communication (ISAC) enables the joint design of communication and sensing functionalities using a shared waveform, but its performance critically depends on the accuracy of the transmitted signal available at the sensing receiver. In practical systems, this signal is obtained by finite blocklength  decoding . This paper develops an analytical framework for quantifying the impact of the associated  decoding uncertainty on sensing performance. We model decoding errors via an equivalent noise formulation that captures their second-order effect on the sensing receiver, and derive the corresponding Fisher Information and Cramér--Rao Bound (CRB). The resultant characterization reveals a fundamental sensing--communication tradeoff governed by the communication rate and decoding reliability.  This shows that the classical finite blocklength reliability results in  a fundamentally new insight concerning ISAC systems: when the decoded communication waveforms are reused as sensing references in a data-aided fashion, the communication reliability boundary also acts as a sensing-information bound.

{We identify and characterize a sharp transition in sensing performance -- referred to as the Tradeoff Cliff -- arising from finite blocklength reliability effects.} This transition separates the regimes of near-ideal sensing performance from those of excessive estimation error as the communication rate approaches capacity. Furthermore, we provide an explicit expression for the critical rate at which this transition occurs, showing its dependence on blocklength, channel dispersion, and target error probability. Monte Carlo simulations under both AWGN and block fading channels validate the theoretical analysis and confirm the trend observed. The results provide design insights for operating ISAC systems with an appropriate reliability margin to avoid severe sensing degradation.
\end{abstract}
%\end{abstract}


\maketitle
%%%%%%%%%%%%%%%%%%%%%%%%%%%%%%%%%%%%%%%%%%%%%%%%%%%%%%%%%%%%%%%%%%%%%%%
%%%%%%%%%%%%%%%%%%%%%%%%%%%%%%%%%%%%%%%%%%%%%%%%%%%%%%%%%%%%%%%%%%%%%%
\section{Introduction}

Integrated sensing and communication (ISAC) has emerged as a key next-generation (NG) paradigm enabling unified wireless communication and environmental sensing using shared spectrum, hardware, and waveforms \cite{Liu2022_ISAC,Liu2018_MUMIMO,Liu2020_ISACFramework,wei2021joint}. Early research primarily focused on radar--communication coexistence and interference management \cite{Zheng2019_Coexistence}, whereas recent developments investigate fully dual-functional waveform design and joint sensing--communication transmission strategies \cite{Hassanien2016_RadarComm,Liu2021_Resource,Wav1}. From an information-theoretic perspective, ISAC tradeoffs have also been studied through rate--distortion, CRB--rate region, and capacity-based formulations \cite{Ding2021_IT_ISAC,Bica2022_Tradeoff,Xiong2023_Fundamental,Yifeng}. However, most existing ISAC frameworks implicitly assume that the sensing processor has access to a perfectly known reference waveform.

This assumption becomes increasingly restrictive in practical communication-centric sensing architectures. In many ISAC systems, particularly those involving small-cell base stations, user equipment (UE), or communication-assisted sensing (CAS), the sensing scheme relies on potentially error-infested decoded communication waveforms, rather than on the  perfectly known reference signal \cite{Dong2024CAS,Masouros2020}. As a result, sensing performance becomes inherently coupled to communication reliability. This issue \textcolor{black}{also arises} in mmWave and THz systems, where hardware impairments such as oscillator phase noise, IQ imbalance, and power-amplifier nonlinearities introduce waveform mismatch and residual sensing uncertainty \cite{Arslan2025_PA,Le2024Hardware,Li2021Mismatched,Ozturk2022,mezghani2022blind}.
\textcolor{black}{In contrast to conventional communication-assisted sensing (CAS) frameworks
\cite{Dong2024CAS,Masouros2020}, which typically assume that the sensing
scheme operates with a decoded communication waveform, the proposed
framework treats finite-blocklength decoding uncertainty as a specific
source of reference-waveform mismatch. More generally, the same
covariance-equivalent formulation is applicable to other sources of
waveform mismatch, including hardware impairments such as oscillator
phase noise, IQ imbalance, and power-amplifier nonlinearities
\cite{Arslan2025_PA,Le2024Hardware,Li2021Mismatched,Ozturk2022,mezghani2022blind}. In
this paper, we focus on the finite-blocklength case, where the mismatch
covariance becomes reliability-dependent and gives rise to the
reliability-induced \emph{Tradeoff Cliff}.}

Emerging NG services such as Ultra-Reliable Low-Latency Communications (URLLC) increasingly rely on short-packet transmissions, where communication reliability is governed by finite blocklength (FBL) theory rather than by the asymptotic Shannon capacity \cite{Polyanskiy2010_FBL}. In this regime, the decoding error probability escalates sharply as the communication rate approaches capacity. For sensing architectures that reuse decoded communication waveforms as sensing references, this implies that decoding unreliability directly results in  sensing degradation and hence also degrades the achievable Cramér--Rao Bound (CRB) \cite{shehab2023finite,Zhang2026}. Similar challenges also arise in CommSense \cite{Commsense1,Commsense2,Tan2021_CommSense} and Wi-Fi fingerprinting systems \cite{Bahl2000_RADAR,Ma2019_WiFi,Zhang2022_DeepWiFi}, where sensing is performed using ambient communication signals without access to an ideal reference waveform.  {Preliminary information-theoretic insights of Commsense were presented in the poster \cite{Khan2026_CommSenseLimits}.}

Recent surveys have identified the characterization of fundamental ISAC performance limits under realistic communication constraints as an important open research direction \cite{Lu2024OpenChallenges}. Nevertheless, most existing analyses either assume asymptotically reliable decoding or model sensing uncertainty as a static impairment independent of communication reliability \cite{Blandino2023,Demirhan2022}. Consequently, the impact of finite blocklength decoding uncertainty on Fisher Information remains insufficiently understood.
\begin{figure*}[htbp]
    \centering
    \includegraphics[width=0.85\linewidth,height=2in]{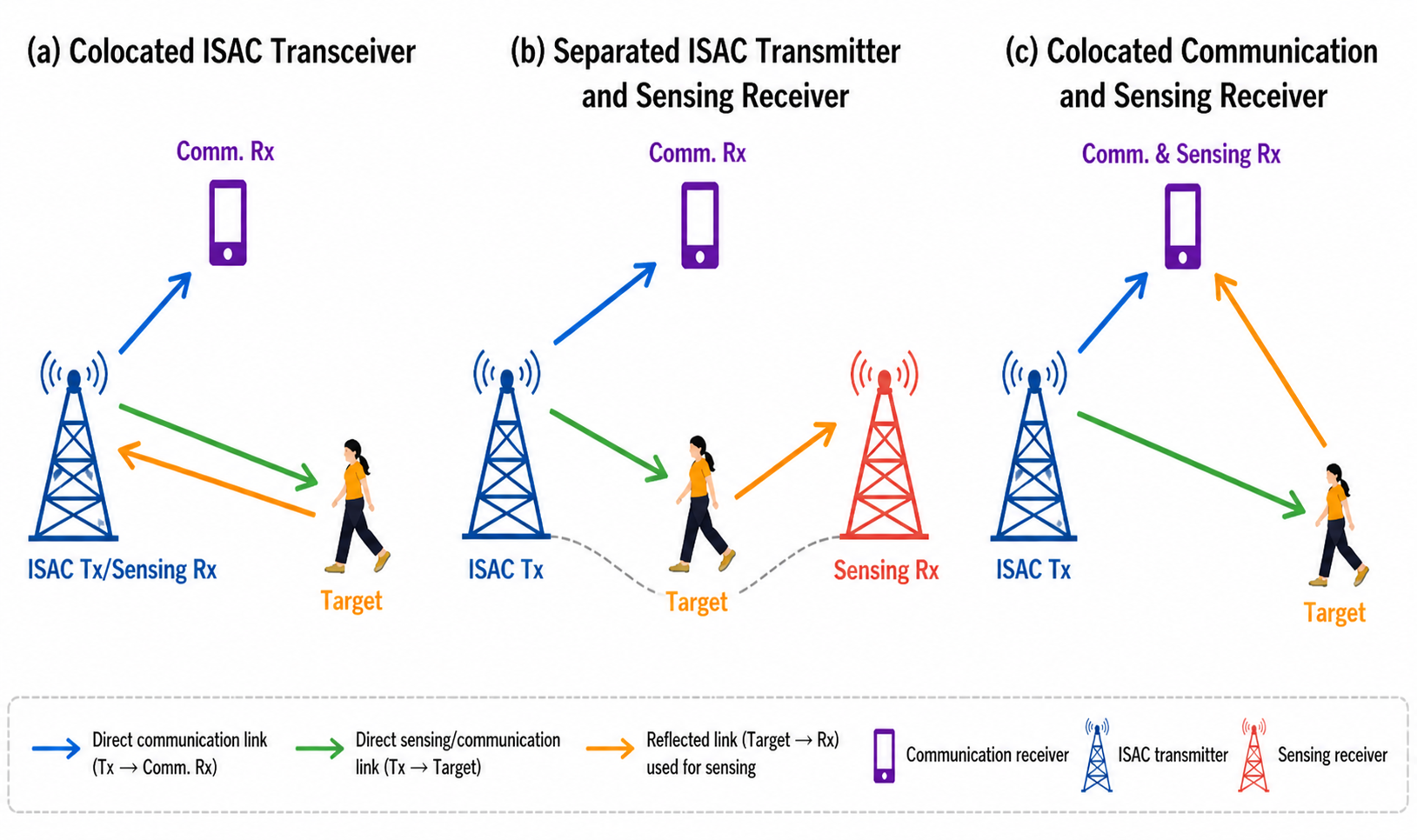}
    \caption{Representative ISAC architectures considered in this work: (a) colocated ISAC transceiver, (b) separated sensing receiver (bistatic sensing), and (c) colocated communication and sensing receiver. The sensing receiver operates using a decoded estimate of the transmitted waveform, introducing reliability-dependent waveform uncertainty.}
    \label{fig:system_model}
\end{figure*}
\textcolor{black}{Recent studies have also investigated communication--sensing tradeoffs
under finite-blocklength communication
\cite{Shen2026TIT,Lin2026TWC}. These works primarily focus on
communication-theoretic tradeoff characterization and resource
optimization under finite-blocklength constraints. In contrast, the
present work addresses a complementary problem by developing a
reliability-aware uncertainty propagation framework that explicitly
quantifies how finite-blocklength decoding uncertainty degrades Fisher
Information and inflates the Cram\'er--Rao Bound through a
covariance-equivalent waveform-mismatch model.}

Motivated by this gap, this paper develops a reliability-aware sensing framework for ISAC systems operating under finite blocklength communication constraints. The novelty of this work does not lie in the finite blocklength reliability transition itself, which is well established in communication theory, but rather in demonstrating how this reliability transition fundamentally propagates into sensing-information degradation when decoded communication waveforms are reused as sensing references. \textcolor{black}{Within the adopted decoded-reference
sensing framework, the proposed analysis establishes the fundamental
reliability limits arising from finite-blocklength decoding uncertainty.} In particular, the proposed framework establishes an explicit analytical relationship between finite blocklength decoding uncertainty and Fisher Information degradation.

The main contributions of this paper are summarized as follows:

\begin{itemize}

\item \textbf{Reliability-Aware Sensing Model:}
We formulate a data decoding-uncertainty-aware ISAC framework in which the sensing receiver operates using decoded communication waveforms rather than perfectly known reference signals.

\item \textbf{Fisher Information and CRB Characterization:}
We derive both the Fisher Information (FI) and the corresponding Cramér--Rao Bound under decoding uncertainty, explicitly quantifying the impact of communication reliability on sensing precision.

\item \textbf{Finite Blocklength Reliability Analysis:}
By harnessing finite blocklength communication theory, we characterize the resultant sensing-performance transition \textcolor{black}{and establish the
fundamental reliability limits within the adopted decoded-reference
sensing framework, explicitly revealing  their} dependence on the communication rate, blocklength, and SNR. 

\item \textbf{Tradeoff Geometry and Scaling Laws:}
We characterize the Pareto boundary of the  communication rate and sensing performance, derive the associated high-curvature sensing transition, and establish the safe operating margin scaling law of
\[
C-R_{\mathrm{safe}}
=
\Theta(n^{-1/2}).
\]

\item \textbf{Reliability-Aware Optimization:}
We formulate a sensing-communication optimization framework that maximizes sensing performance subject to communication reliability and Quality-of-Service constraints.

\end{itemize}

Table~\ref{tab:isac_comparison} summarizes the differences between classical ISAC formulations and the proposed reliability-aware framework.
In contrast to classical CRB--rate tradeoff formulations such as \cite{Xiong2023_Fundamental}, the proposed framework explicitly incorporates finite blocklength decoding reliability into the sensing-information characterization \textcolor{black}{ within the adopted
decoded-reference based sensing framework}.
\begin{table*}[t]
\centering
\caption{Comparison with Representative ISAC and Reliability-Aware Sensing Literature}
\label{tab:isac_comparison}
\renewcommand{\arraystretch}{1.15}
\setlength{\tabcolsep}{4pt}

\begin{tabular}{|p{4.2cm}|c|c|c|c|c|c|c|}
\hline

\textbf{Reference}
&
\shortstack{\textbf{Decoded /}\\ \textbf{Imperfect}\\ \textbf{Waveforms}}
&
\shortstack{\textbf{Finite}\\ \textbf{Blocklength}}
&
\shortstack{\textbf{Reliability-}\\ \textbf{Aware}\\ \textbf{Sensing}}
&
\shortstack{\textbf{FI / CRB}\\ \textbf{Analysis}}
&
\shortstack{\textbf{Tradeoff}\\ \textbf{Cliff}}
&
\shortstack{\textbf{Safe Margin}\\ \textbf{Scaling}}
&
\shortstack{\textcolor{black}{\textbf{Reliability-Aware}}\\ \textcolor{black}{\textbf{FI/CRB}}}
\\
\hline

Liu \emph{et al.}~\cite{Liu2018_MUMIMO}
&

&

&

&

&

&

&
\\
\hline

Liu \emph{et al.}~\cite{Liu2020_ISACFramework}
&

&

&

&

&

&

&
\\
\hline

Xiong \emph{et al.}~\cite{Xiong2023_Fundamental}
&

&

&

&
\checkmark
&

&

&
\\
\hline

Hassanien \emph{et al.}~\cite{Hassanien2016_RadarComm}
&

&

&

&

&

&

&
\\
\hline

Dong \emph{et al.}~\cite{Dong2024CAS}
&
\checkmark
&

&

&

&

&

&
\\
\hline

Blandino \emph{et al.}~\cite{Blandino2023}
&
\checkmark
&

&

&

&

&

&
\\
\hline

Polyanskiy \emph{et al.}~\cite{Polyanskiy2010_FBL}
&

&
\checkmark
&

&

&

&

&
\\
\hline

Tan \emph{et al.}~\cite{Tan2021_CommSense}
&
\checkmark
&

&

&

&

&

&
\\
\hline

\textcolor{black}{Lin \emph{et al.}~\cite{Lin2026TWC}}
&

&
\checkmark
&
\checkmark
&

&

&

&
\\
\hline
\textcolor{black}{Shen \emph{et al.}~\cite{Shen2026TIT}}
&

&
\checkmark
&

&

&

&

&
\\
\hline
\textbf{Proposed Framework}
&
\checkmark
&
\checkmark
&
\checkmark
&
\checkmark
&
\checkmark
&
\checkmark
&
\checkmark
\\
\hline

\end{tabular}
\end{table*}

The remainder of this paper is organized as follows. Section~\ref{sec2} presents the system model and decoding uncertainty framework. Section~\ref{sec3} derives the Fisher Information under decoding uncertainty, while Section~\ref{sec4} characterizes the resultant sensing--communication tradeoff geometry. Section~\ref{sec6} develops the associated reliability-aware optimization framework. Our numerical results are presented in Section~\ref{sec7}, followed by our conclusions in Section~\ref{sec8}.
%%%%%%%%%%%%%%%%%%%%%%%%%%%%%%%%%%%%%%%%%%%%%%%%%%%%%%%%%%%%%%%%%%%%%%%%%
%%%%%%%%%%%%%%%%%%%%%%%%%%%%%%%%%%%%%%%%%%%%%%%%%%%%%%%%%%%%%%%%%%%%%%%%%%%
\section{System Model}\label{sec2}

We consider a point-to-point integrated sensing and communication (ISAC) system in which a dual-functional transmitter simultaneously performs wireless communication and target sensing using a shared waveform \(x\). The sensing receiver may be co-located with the transmitter, co-located with the communication receiver, or spatially separated, as illustrated in Fig.~\ref{fig:system_model}. In contrast to conventional ISAC formulations that assume perfect waveform knowledge, the sensing receiver operates using a potentially error-infested decoded estimate of the transmitted signal, denoted by \(\hat{x}\). Consequently, the sensing performance becomes coupled to communication reliability.

\textcolor{black}{\textbf{Applicability of the sensing architecture:} Figure 1 illustrates representative ISAC architectures in which the sensing processor may obtain the reference waveform through different paths. The proposed framework is applicable when the waveform available at the sensing processor is an imperfect representation of the actually transmitted waveform. This includes receiver-side, passive, communication-assisted, and third-party sensing architectures in which the reference waveform may be obtained through a decoded communication signal. In the finite-blocklength setting considered in this work, decoding errors introduce uncertainty between the transmitted waveform and the decoded reference, resulting in a reliability-dependent reference mismatch.
\\
Importantly, the framework does not imply that a decoded reference is required in every ISAC architecture. In conventional monostatic or transmitter-side ISAC systems, where the sensing processor is colocated with the transmitter and has direct access to the transmitted baseband samples, the transmitted waveform is known and the decoding-induced reference mismatch is absent. Such a system corresponds to the zero-mismatch special case of the proposed framework. The framework can also represent other practical sources of reference uncertainty, such as synchronization errors, calibration inaccuracies, channel-estimation errors, and hardware impairments, by appropriately modeling the corresponding mismatch covariance.
\\
The specific \textbf{Tradeoff Cliff} investigated in this work, however, results from the coupling between finite-blocklength decoding reliability and the reference-waveform mismatch covariance. Hence, while the covariance-based formulation is applicable more broadly to imperfect-reference sensing, the reliability-dependent Tradeoff Cliff derived in this paper is specifically associated with decoding-induced waveform uncertainty.
}
\subsection{Signal and Communication Model}

The received sensing signal is modeled as
\begin{equation}
y=hx+w,
\end{equation}
where:
\begin{itemize}
\item \(h\) denotes the sensing parameter or target-dependent channel coefficient,
\item \(w\sim\mathcal{CN}(0,N_0)\) is additive complex Gaussian noise,
\item and \(x\) is the transmitted waveform satisfying
\[
\mathbb{E}[|x|^2]=P.
\]
\end{itemize}

Communication relies on the finite blocklength \(n\) at transmission rate \(R\) (bits/channel use). Under finite blocklength (FBL) communication, the decoding error probability is approximated as \cite{Polyanskiy2010_FBL,shehab2023finite}
\begin{equation}
P_e(R,n)
\approx
Q\!\left(
\sqrt{\frac{n}{V}}
(C-R)\ln2
\right),
\end{equation}
where \(C\) and \(V\) denote the Shannon capacity and channel dispersion, respectively.

\subsection{Decoding Uncertainty Model}

The sensing receiver utilizes the FBL decoded communication waveform as the sensing reference. We model the decoded signal as
\begin{equation}\label{eq3}
\hat{x}=x+e,
\end{equation}
where \(e\) denotes the effective aggregate decoding distortion. We assume
\[
e\sim\mathcal{CN}[0,\sigma_e^2(R,n)],
\]
with \(e\)  being independent of both \(x\) and \(w\).

\textcolor{black}{The variance $\sigma_e^2(R,n)$ captures the residual uncertainty induced
by finite-blocklength decoding and increases with the communication
unreliability. As shown by the covariance decomposition developed in
Appendix~\ref{app:model}, the aggregate residual covariance is obtained
by averaging over the successful and unsuccessful decoding events.
Since the residual covariance under successful decoding is negligible,
the aggregate residual uncertainty scales with the decoding error
probability. Let
\[
P_e=P_e(R,n)
\]
denote the finite-blocklength decoding error probability. This yields
the equivalent covariance model
\begin{equation}
\sigma_e^2(R,n)
=
\kappa P_e(R,n),
\end{equation}
where $\kappa>0$ denotes the average per-symbol residual distortion
energy conditioned on a block decoding failure and therefore depends on
the coding scheme, modulation format, decoder implementation, and
operating conditions. For a given communication system, the value of $\kappa$ can be obtained
through decoder simulations or experimental measurements by estimating
the average residual distortion energy conditioned on unsuccessful block
decoding.} 

\textcolor{black}{The proposed formulation models the aggregate second-order impact of
decoding uncertainty on sensing rather than the exact symbol-level
structure of decoding failures. Appendix~\ref{app:model} provides the
complete covariance-based derivation of the equivalent Gaussian
decoding-distortion model together with a discussion of its assumptions,
scope, and limitations.}
\subsection{Sensing Model}
The monostatic sensing receiver probes the environment and receives an echo signal reflected from the target. Traditionally, radar processing performs matched filtering using the perfectly known reference waveform $x$. However, in this uncertainty-aware model, the receiver is constrained to use the imperfect estimate $\hat{x}$ as the `reference' signal. This mismatch between the actual echo (generated by $x$) and the reference (the noisy $\hat{x}$) introduces a bias and increased variance in the estimation of sensing parameters, such as delay and Doppler shift.
%%%%%%%%%%%%%%%%%%%%%%%%%%%%%%%%%%%%%%%%%%%%%%%%%%%%%%%%%%%%%%%%%%%%%%%%%
%%%%%%%%%%%%%%%%%%%%%%%%%%%%%%%%%%%%%%%%%%%%%%%%%%%%%%%%%%%%%%%%%%%%%%%%%
\section{Fisher Information with Decoding Uncertainty}\label{sec3}
Fisher \textcolor{black}{I}nformation and \textcolor{black}{the associated} CRB-based formulations are widely used for characterizing the sensing performance and the fundamental sensing--communication tradeoffs in ISAC systems \cite{Fortunati2020_CRB_MIMO_CRB,Xiong2023_Fundamental,Pucci2025_FI_ISAC,hua2023mimo,CRB1,PER1}. \textcolor{black}{Throughout this paper, the sensing performance metric
$S(R,n)$ denotes the normalized Fisher Information.
For an observation consisting of $n$ coherently processed samples,
the total Fisher Information satisfies
\[
I_{\rm tot}(R,n)\approx n\,S(R,n),
\]
where the linear dependence on $n$ represents the classical coherent
integration gain. The normalization adopted in this paper removes this
well-known scaling so that the analysis isolates the effect of
finite-blocklength decoding reliability on sensing performance.
We emphasize that this normalization is a deliberate definitional choice rather than an incidental simplification: it separates the classical coherent integration gain — which scales $I_{\rm tot}(R,n)$ uniformly across all operating rates $R$ — from the reliability-dependent covariance inflation governed by $P_e(R,n)$, which reshapes the relative sensing performance across $R$. Consequently, the blocklength $n$ influences sensing performance through two distinct mechanisms: (i) a uniform scaling of absolute Fisher Information via integration gain, and (ii) a reliability-driven sharpening of the Tradeoff Cliff (cf. Theorem 3), which is a statement about the shape of $S(R,n)$.
}
\subsection{Equivalent Model and Effective Noise}
To analyze the sensing limits, we consider the observation at the collocated receiver. Given the decoded reference $\hat{x} = x + e$, the received echo of $y = hx + w$ can be reformulated to isolate the available reference:
\begin{equation}
    y = h\hat{x} + \tilde{w}, \quad \text{where } \tilde{w} = w - he.
\end{equation}
Here, $w \sim \mathcal{CN}(0, N_0)$ represents thermal noise and $e \sim \mathcal{CN}(0, \sigma_e^2)$ represents the decoding uncertainty. Assuming that $w$ and $e$ are independent, the effective noise $\tilde{w}$ follows a complex Gaussian distribution $\tilde{w} \sim \mathcal{CN}(0, \sigma_{\tilde{w}}^2)$ with variance:
\begin{equation}
    \sigma_{\tilde{w}}^2 = \mathbb{E}[|\tilde{w}|^2] = N_0 + |h|^2 \sigma_e^2(R, n).
\end{equation}

\subsection{Approximate Fisher Information under Decoding Covariance Inflation}
The effective observation model of $
y=h\hat{x}+\tilde w,
\qquad
\tilde w=w-he,
$
induces a heteroscedastic likelihood because the effective noise variance depends on the unknown sensing parameter $h$. Consequently, the resultant Fisher Information contains both mean-gradient and variance-gradient contributions. To obtain a tractable closed-form characterization, we adopt the equivalent second-order covariance framework developed in Appendix~\ref{app:model} and neglect higher-order residual correlation terms arising from the coupling between $\hat{x}$ and the decoding distortion.

\textcolor{black}{Since the decoded waveform $\hat{x}$ is itself a random quantity resulting
from the finite-blocklength decoding process, the sensing performance is
naturally characterized by the expected Fisher Information obtained by
first conditioning on a fixed decoded-reference realization $\hat{x}$,
and then averaging over the distribution of $\hat{x}$ induced by
finite-blocklength decoding:
\begin{equation}
    I(h) = \mathbb{E}_{\hat{x}}\big[I(h|\hat{x})\big].
    \label{eq:expected_FI}
\end{equation}
This two-step construction --- first the conditional Fisher Information
$I(h|\hat{x})$ for a fixed $\hat{x}$, then the expectation over
$\hat{x}$ --- is made explicit in Appendix B. The physical meaning of
this metric is that it quantifies the expected sensing information
available to a decoded-reference ISAC receiver before the realization of
the decoding outcome is known, and therefore characterizes its average
sensing capability under finite-blocklength communication. This
interpretation is consistent with classical estimation theory~\cite{kay1993}
and recent communication-centric ISAC frameworks that evaluate sensing
performance by averaging over random communication waveform
realizations~\cite{Liu2025RandomISAC,Lu2025RandomSignals}. Accordingly, unless
otherwise stated, all Fisher Information expressions derived in this paper
refer to this expected Fisher Information.
}

\begin{theorem}
Under the equivalent covariance-based decoding uncertainty model developed in Appendix~\ref{app:model}, and assuming moderate decoding uncertainty for ensuring that higher-order residual correlation terms are negligible, the \textcolor{black}{expected} Fisher Information (FI) for estimating the complex channel gain $h$ admits the approximation
\begin{equation}\label{eq8}
I(h)
\approx
\frac{\mathbb{E}[|\hat{x}|^2]}
{N_0+|h|^2\sigma_e^2}
+
\frac{|h|^2\sigma_e^4}
{(N_0+|h|^2\sigma_e^2)^2}.
\end{equation}
\end{theorem}

\begin{proof}
See Appendix \ref{app:FI_proof}.
\end{proof}
The first term corresponds to the conventional mean-gradient contribution scaled by the inflated effective covariance, while the second term captures the variance-gradient contribution imposed by the parameter-dependent decoding uncertainty.

The analysis extends naturally to multi-parameter estimation as seen in the next subsection.
%; see Appendix \ref{app:vector_fi} for the vector Fisher Information formulation.
\subsubsection{High-SNR Approximation}
In the practically relevant regime where the equivalent decoding covariance remains moderate relative to the signal power, i.e.,
$
\sigma_e^2 \ll P,
$
the dominant sensing degradation mechanism arises through effective covariance inflation. Accordingly,
\[\mathbb{E}[|\hat x|^2]
=
P+\sigma_e^2
\approx
P.
\]
Consequently, using Appendix Subsection \ref{AG}, the sensing metric obeys $S(R, n) \approx \frac{P}{N_0 + |h|^2 \sigma_e^2(R, n)}$.
\subsection{\color{black}Extension to Multi-Parameter Estimation}
{\color{black}
The scalar Fisher Information analysis captures how finite-blocklength decoding uncertainty impacts sensing. However, practical ISAC systems require the joint estimation of multiple target parameters.

Let the unknown parameter vector be
\begin{equation}
\boldsymbol{\theta} = \left[ \tau,\; \nu,\; \Re\{h\},\; \Im\{h\} \right]^T,
\end{equation}
where $\tau$, $\nu$, and $h$ denote the target delay, Doppler shift, and complex channel gain, respectively.

Under the decoded-reference observation model, the received signal depends on $\boldsymbol{\theta}$ through both its mean and effective covariance. Consequently, the sensing performance is characterized by the vector Fisher Information Matrix (FIM).
}
\begin{theorem}[Tradeoff Cliff under Multi-Parameter Estimation]\label{mulparThem}
{\color{black}
Consider the parameter vector $\boldsymbol{\theta} = \left[ \tau,\; \nu,\; \Re\{h\},\; \Im\{h\} \right]^T$. Under the equivalent covariance-inflation framework, the Fisher Information Matrix (FIM) is approximated by

\begin{equation}
\mathbf{I}(\boldsymbol{\theta}) \approx \mathbf{I}_{\rm mean} + \mathbf{I}_{\rm var},
\end{equation}

where $\mathbf{I}_{\rm mean}$ and $\mathbf{I}_{\rm var}$ arise from the parameter dependence of the observation mean and covariance. Under moderate decoding uncertainty ($\sigma_e^2 \ll P$), this simplifies to

\begin{equation}
\mathbf{I}(\boldsymbol{\theta}) \approx \frac{1}{N_0+|\alpha|^2\sigma_e^2(R,n)}
\mathbf{I}^{(0)}(\boldsymbol{\theta}),
\end{equation}

where $\mathbf{I}^{(0)}(\boldsymbol{\theta})$ is the conventional FIM under a perfectly known reference waveform. Consequently, the CRBs for all jointly estimated sensing parameters exhibit the same reliability-induced Tradeoff Cliff as the communication rate approaches the finite-blocklength boundary.}
\end{theorem}

\begin{IEEEproof}
{\color{black}See Appendix~C.}
\end{IEEEproof}
{\color{black}
Theorem \ref{mulparThem} demonstrates that the proposed reliability-aware sensing framework naturally extends to practical multi-parameter estimation. Although the off-diagonal entries of the Fisher Information Matrix modify the coupling among delay, Doppler, and channel-gain estimates, the reliability-dependent covariance inflation scales the entire matrix. Consequently, the Tradeoff Cliff is a fundamental property of decoded-reference based sensing and is not restricted to the scalar case.}

\textcolor{black}{
\begin{rem}
    The proposed equivalent covariance formulation may also be interpreted
as a tractable second-order approximation to the waveform mismatch
considered in the misspecified Cramér--Rao Bound (MCRB) literature
\cite{Vuong1986,Richmond2015,Fortunati2016}. A complete MCRB treatment for
decoded-reference based sensing under finite-blocklength communication is an
interesting direction for future research, but it is beyond the scope of this work, given the strict page limit.
\end{rem}
}
%%%%%%%%%%%%%%%%%%%%%%%%%%%%%%%%%%%%%%%%%%%%%%%%%%%%%%%%%%%%%%%%%%%%%%%%%%
%%%%%%%%%%%%%%%%%%%%%%%%%%%%%%%%%%%%%%%%%%%%%%%%%%%%%%%%%%%%%%%%%%%%%%%%%%%%%
\section{Tradeoff Cliff Characterization}\label{sec4}
This section formally characterizes the sensing--communication tradeoff imposed by decoding uncertainty under FBL communication. We demonstrate that the sensing performance degrades monotonically as the communication rate increases, creating a sharp transition — the \emph{Tradeoff Cliff} — whose geometry is governed by FBL reliability theory. 

\subsection{Pareto-Optimal Reliability-Driven Tradeoff Boundary}  
\begin{lemma}\label{thm2}
The Pareto boundary defining the tradeoff between the communication rate $R$ and the sensing performance $S$ is:
\begin{equation}\label{eq12}
    S(R, n) = \frac{P}{N_0 + \xi \sigma_e^2(R, n)},
\end{equation}
where $\xi = |h|^2$ is the target-dependent scaling factor.
\end{lemma}
\begin{proof}
To establish the Pareto boundary, we show that $S$ is monotonically non-increasing with $R$. \textcolor{black}{Recalling that $S(R,n)$ denotes the normalized expected Fisher Information $I(h)$ defined in Section III}, we rely on the normal FBL approximation for the decoding error probability \cite{Polyanskiy2010_FBL}: $P_e(R, n) \approx Q\left( \frac{C - R}{\sqrt{V/n}} \ln 2 \right)$. Upon modeling the uncertainty variance as proportional to this error rate, $\sigma_e^2(R, n) = \kappa P_e(R, n)$ for $\kappa > 0$, we apply the chain rule: $\frac{\partial S}{\partial R} = \frac{\partial S}{\partial \sigma_e^2} \frac{\partial \sigma_e^2}{\partial P_e} \frac{\partial P_e}{\partial R}$. 

Since sensing precision decreases as uncertainty grows ($\frac{\partial S}{\partial \sigma_e^2} < 0$), uncertainty scales linearly with the error ($\frac{\partial \sigma_e^2}{\partial P_e} > 0$), and the error probability increases with rate ($\frac{\partial P_e}{\partial R} > 0$), it follows that $\frac{\partial S}{\partial R} < 0$. Thus, increasing the communication rate strictly degrades the sensing performance, confirming \eqref{eq12} as the Pareto-optimal boundary.
\end{proof}

\begin{figure}[htbp]
    \centering
    \includegraphics[width=0.95\linewidth]{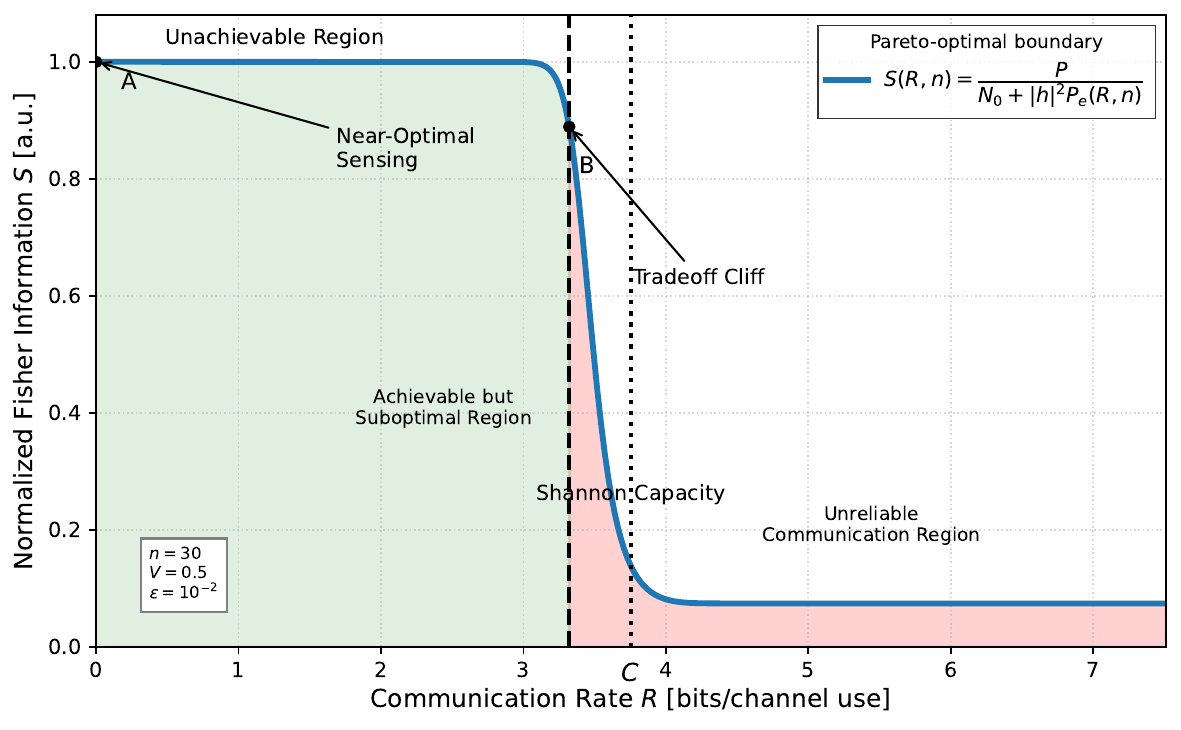}
    \caption{\textbf{Pareto Boundary of ISAC Systems.} The fundamental tradeoff between the rate $R$ and sensing performance $S$. At low rates (Point A), sensing is near-optimal. As $R \to C$ (Point B), FBL effects cause a ``tradeoff cliff'' where sensing drops precipitously.}
    \label{fig:pareto_tradeoff}
\end{figure}

As illustrated in Fig. \ref{fig:pareto_tradeoff}, the system operates in three distinct regimes:
\begin{enumerate}
    \item \textbf{Near-Optimal Sensing ($R \ll C$):} Highly reliable communication ($P_e \approx 0$) yields negligible sensing uncertainty, allowing precision near the theoretical limit $P/N_0$.
    \item \textbf{Tradeoff Cliff ($R \to C$):} The FBL regime fails to suppress errors which results in escalating sensing noise.
    \item \textbf{Unreliable Communication ($R > C$):} High data outage yields heavily distorted references, providing inadequate sensing.
\end{enumerate}

\subsection{Tradeoff Cliff Threshold}
\begin{lemma}[Tradeoff Cliff Threshold]\label{prop:cliff_rate}
For a target decoding error probability $\epsilon \in (0,1)$, the operational communication rate marking the Tradeoff Cliff is:
\begin{equation}
R_{\text{cliff}} = C - \sqrt{\frac{V}{n}} \frac{Q^{-1}(\epsilon)}{\ln 2}.
\end{equation}
\end{lemma}
\begin{proof}
The cliff $R_{\text{cliff}}$ is operationally defined where $P_e(R_{\text{cliff}}, n) = \epsilon$. Substituting the FBL normal approximation yields $\epsilon = Q\left( \sqrt{n/V} (C - R_{\text{cliff}})\ln 2 \right)$. Applying the inverse $Q$-function and solving for $R_{\text{cliff}}$ yields the result. Because $P_e(R,n)$ increases monotonically with $R$, rates beyond $R_{\text{cliff}}$ suffer from sharp sensing degradation.
\end{proof}

This expression shows that  the cliff occurs at a capacity margin scaling formulated as $\mathcal{O}(1/\sqrt{n})$. Shorter blocklengths widen this gap, causing earlier sensing degradation.

\subsection{Cliff Curvature and Sharpness Scaling}\label{sub4c}
To characterize the tradeoff's geometry, we evaluate the derivatives of $S(R,n)$ with respect to $R$. Let $a=\sqrt{n/V}\ln2$.

\begin{theorem}[Curvature and Sharpness]\label{thm4}
The first and second derivatives of the sensing performance are:
\begin{align}
    \frac{dS}{dR} &= -\frac{P\xi\kappa a}{\sqrt{2\pi}} \frac{e^{-\frac{a^2(C-R)^2}{2}}}{\left(N_0+\xi\kappa P_e(R,n)\right)^2}, \\
    \frac{d^2S}{dR^2} &= \textcolor{black}{-} \frac{P\xi\kappa a^3(C-R)}{\sqrt{2\pi}} \frac{e^{-\frac{a^2(C-R)^2}{2}}}{\left(N_0+\xi\kappa P_e(R,n)\right)^2} \nonumber\\
    &~~~~~\textcolor{black}{+} \frac{P(\xi\kappa)^2a^2}{\pi} \frac{e^{-a^2(C-R)^2}}{\left(N_0+\xi\kappa P_e(R,n)\right)^3}. \nonumber
\end{align}
Consequently, the tradeoff exhibits an inflection region near $R\approx C$. The maximum slope gradient scales as $|\frac{dS}{dR}|_{\max} = \mathcal{O}(\sqrt{n})$, implying that the cliff approaches a step-like transition at $C$ as $n\to\infty$.
\end{theorem}
\begin{proof}
Substituting $P_e(R,n) = Q\left[a(C-R)\right]$ into $S(R,n)$ and applying standard differentiation rules yields the derivatives. Since $a \propto \sqrt{n}$, the slope's gradient is directly proportional to $\sqrt{n}$, confirming the sharpening of the cliff for higher blocklengths.
\end{proof}
\textbf{Remark:}
The dominant contribution arises near the reliability transition
region where the Gaussian density term associated with
\(Q'(x)\) attains its maximum magnitude at \(x \approx 0\),
corresponding to \(R \approx C\). Consequently,
\[
\left|\frac{dS}{dR}\right|_{\max}
\propto
a
=
\sqrt{\frac{n}{V}}\ln 2,
\]
which quantifies the \(O(\sqrt{n})\) cliff sharpness.
\subsection{Safe Operating Margin Scaling}
We next quantify the required margin below capacity to maintain stable sensing.

\begin{theorem}[Safe Operating Margin]\label{thm_safe_margin}
To guarantee a minimum sensing performance $S(R,n) \geq S_{\mathrm{th}}$, the communication rate must not exceed $R_{\mathrm{safe}} = C-\Delta_n$, where the reliability margin $\Delta_n$ is:
\begin{equation}
\Delta_n = \frac{\sqrt{V}}{\sqrt{n}\ln2} \,Q^{-1}\!\left( \frac{1}{\xi\kappa} \left( \frac{P}{S_{\mathrm{th}}}-N_0 \right) \right) = \Theta\!\left(\frac{1}{\sqrt n}\right).
\label{eq:safe_margin}
\end{equation}
\end{theorem}
\begin{proof}
The constraint $S(R,n) \geq S_{\mathrm{th}}$ requires $P_e(R,n) \leq \frac{1}{\xi\kappa}\left(\frac{P}{S_{\mathrm{th}}}-N_0\right)$. Substituting the FBL approximation, applying the inverse $Q$-function, and isolating $\Delta_n = C - R_{\mathrm{safe}}$ yields \eqref{eq:safe_margin}. As the other terms are independent of $n$, the margin strictly scales as $\Theta(n^{-1/2})$.
\end{proof}

\textbf{Remark:} Short-packet ISAC systems have a reliability margin scaling as $\mathcal{O}(1/\sqrt{n})$ below capacity to avoid severe sensing degradation, establishing them as being fundamentally reliability-limited.

\subsection{Special Cases}

\textbf{The Ideal Case ($P_e = 0$):} Under error-free decoding, uncertainty vanishes ($\sigma_e^2 = 0$). Sensing performance hits a rate-independent baseline, $S_{\text{ideal}} = P/N_0$. 

\begin{lemma}[Exponential Uncertainty Decay]
If uncertainty decays exponentially with rate (e.g., in source-coding frameworks) such that $\sigma_e^2(R) = \kappa 2^{-R}$, then $S(R) = P/(N_0 + \xi \kappa 2^{-R})$.
\end{lemma}

\begin{figure}[htbp]
    \centering
    \includegraphics[width=0.95\linewidth]{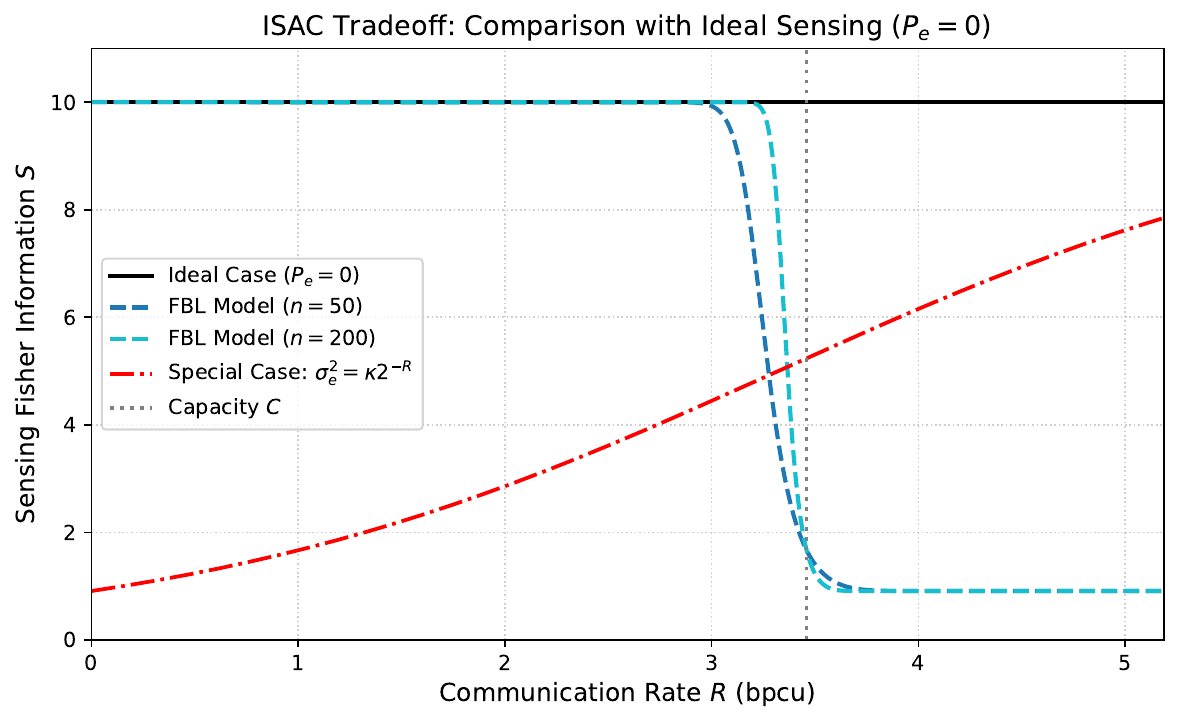}
    \caption{Sensing performance $S$ vs. rate $R$, comparing the ideal baseline ($P_e=0$), the exponential special case, and the FBL model across various blocklengths $n$.}
    \label{fig:special_case_plot}
\end{figure}

\textbf{Remark:} As seen in Fig. \ref{fig:special_case_plot}, the exponential model curves upward, asymptotically saturating at $P/N_0$ as $R \to \infty$. This suits source-limited scenarios, where throughput refines waveform knowledge. Conversely, the FBL model captures the physical reality of noise-limited channels: it maintains near-ideal performance before suffering a rapid collapse at the capacity limit, acting increasingly like a step function for higher blocklengths.
%%%%%%%%%%%%%%%%%%%%%%%%%%%%%%**********%%%%%%%%%%%%%%%%%%%%%%%%%%%%%%%%%%%%%%%%%%%
%%%%%%%%%%%%%%%%%%%%%%%%%%%%%%%%%%%%%%%%%%%%%%%%%%%%%%%%%%%%%%%%%%%%%%%%%%%%%%%%%%%%
\section{Reliability-Aware Operating Point Design}
\label{sec6}

The Tradeoff Cliff characterized in Section~\ref{sec4} fundamentally changes the character of ISAC systems operating under FBL communication constraints. In particular, sensing performance becomes highly sensitive to communication-rate selection near the finite blocklength reliability boundary, where decoding-induced covariance inflation grows rapidly.

In contrast to  conventional ISAC formulations that primarily interpret sensing--communication coupling through explicit resource partitioning, the proposed framework reveals an intrinsic reliability-induced operating constraint: the communication waveform itself becomes an increasingly unreliable sensing reference as the communication rate approaches capacity.

This section characterizes the resultant reliability-aware operating region and investigates the impact of communication Quality-of-Service (QoS) constraints on the achievable sensing performance.

\subsection{Reliability-Constrained Operating Region}

From Theorem~\ref{thm_safe_margin}, maintaining a minimum sensing performance level of 
\[
S(R,n)\geq S_{\mathrm{th}}
\]
requires the communication rate to satisfy
\begin{equation}
R
\leq
R_{\mathrm{safe}}
=
C-\Delta_n,
\label{eq:reliable_region}
\end{equation}
where
\begin{equation}
\Delta_n
=
\frac{\sqrt{V}}{\sqrt{n}\ln2}
Q^{-1}\!\left(
\frac{1}{\xi\kappa}
\left(
\frac{P}{S_{\mathrm{th}}}-N_0
\right)
\right)
\label{eq:delta_safe}
\end{equation}
denotes the finite blocklength reliability margin.

Equation~\eqref{eq:reliable_region} establishes that the feasible sensing region is fundamentally reliability-limited. As the communication rate approaches capacity, the decoding uncertainty increases sharply, inflating the equivalent sensing covariance and causing rapid Fisher Information degradation.

Consequently, reliable sensing operation requires maintaining a certain operating margin below the communication reliability bound.

\subsection{Operating Geometry Near the Tradeoff Cliff}

The curvature analysis developed in Theorem~\ref{thm4} shows that the sensing-performance surface becomes increasingly ill-conditioned near the Tradeoff Cliff. In particular, having
\[
\left|
\frac{dS}{dR}
\right|_{\max}
=
\mathcal{O}(\sqrt n),
\]
implies that the sensing sensitivity to communication-rate variations grows with the communication blocklength.

Accordingly, the operating region may be divided into three distinct regimes:

\begin{itemize}

\item \textbf{Reliable Regime $(R\ll C-\Delta_n)$:}

Decoding uncertainty remains negligible, and sensing performance approaches the ideal covariance-limited behavior
\[
S
\approx
\frac{P}{N_0}.
\]

\item \textbf{Transition Regime $(R\approx C-\Delta_n)$:}

The finite blocklength reliability transition induces rapid covariance inflation, causing the Fisher Information to become highly sensitive to communication-rate variations.

\item \textbf{Reliability-Limited Regime $(R\to C)$:}

The decoded waveform becomes an unreliable sensing reference due to excessive decoding uncertainty, producing severe sensing-information degradation.

\end{itemize}

These regimes collectively define the reliability-aware operating region  of decoded-reference based  data-aided ISAC systems.

\subsection{QoS-Constrained Feasible Region}

We next investigate the impact of communication QoS constraints on the achievable sensing performance.

Consider the sensing optimization problem of
\begin{align}
&\min_{R,P}
\quad
\mathrm{CRB}(R,P)
=
\frac{
N_0+\xi\kappa P_e(R,n)
}{P}
\nonumber\\
&\text{s.t.}
\quad
R\geq R_{\min},
\qquad
0\leq P\leq P_{\max}.
\label{eq:qos_problem}
\end{align}

\begin{Corollary}[Reliability-Constrained Boundary Solution]
\label{cor_boundary}
The sensing-optimal operating point satisfying \eqref{eq:qos_problem} is
\[
R^*=R_{\min},
\qquad
P^*=P_{\max}.
\]
\end{Corollary}

\begin{proof}
The CRB increases monotonically with the communication rate through the decoding error probability $P_e(R,n)$ and decreases monotonically with the transmit power $P$. Consequently, minimizing the sensing error requires operating at the lowest feasible communication rate and highest feasible transmit power.
\end{proof}

Corollary~\ref{cor_boundary} reveals that the achievable sensing performance is fundamentally constrained by the minimum communication QoS requirement. As $R_{\min}$ approaches the finite blocklength reliability boundary, the feasible sensing region contracts rapidly due to decoding-induced covariance inflation.

In particular, if
\[
R_{\min}
\geq
R_{\mathrm{safe}},
\]
then the system is forced to operate inside the Tradeoff Cliff -- transition region, where small increases in communication throughput imposes disproportionately high sensing degradation.

\subsection{Reliability-Aware Feasible Operating Region}

\begin{figure}[htbp]
    \centering
    \includegraphics[width=0.9\linewidth]{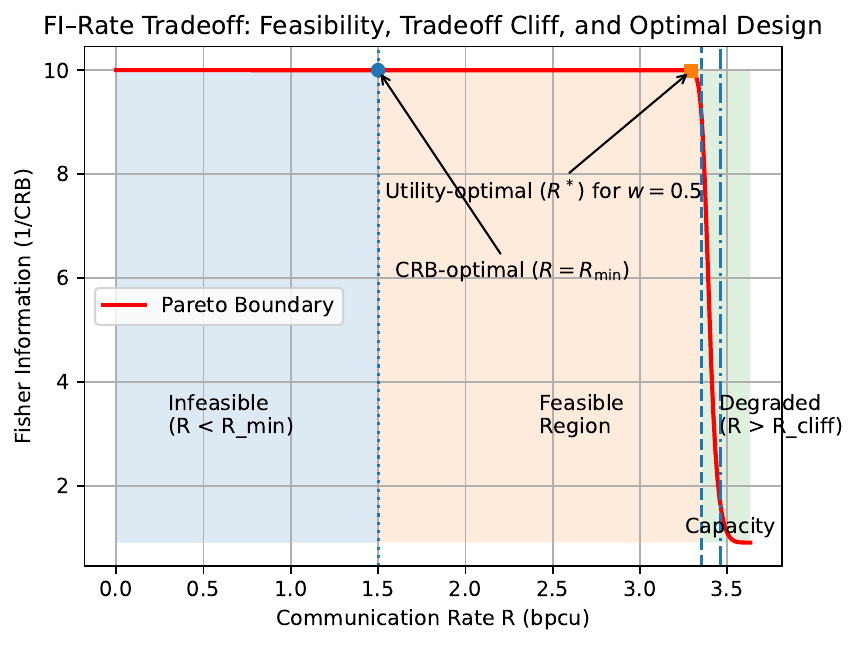}
    \caption{Operating region under decoding uncertainty. $R^*$ balances throughput and fidelity below $R_{\mathrm{cliff}}$, while $R_{\min}$ sets the CRB-optimal boundary. Parameters: $P=1, N_0=0.1, n=500, V=0.5, \kappa=1, \epsilon=10^{-2}$.}
    \label{fig:optimization_region}
\end{figure}
The resultant operating region is illustrated in Fig.~\ref{fig:optimization_region}, where the feasible sensing region is determined jointly by:
\begin{itemize}
\item the finite blocklength reliability boundary,
\item the communication QoS requirement,
\item and the sensing covariance constraint.
\end{itemize}

Several important observations follow:

\begin{itemize}

\item \textbf{Reliability Boundary:}
The Tradeoff Cliff acts as a practical sensing bound separating stable and reliability-limited operating regions.

\item \textbf{Reliability Margin:}
Maintaining a sufficiently high margin relative to capacity substantially improves sensing performance by suppressing decoding-induced covariance inflation.

\item \textbf{QoS-Driven Feasibility Loss:}
Aggressive near-capacity communication-rate requirements progressively shrink the feasible sensing region and may force operation into the instability region near the Tradeoff Cliff.

\end{itemize}

The proposed framework therefore demonstrates that data-aided decoded-reference based ISAC systems are fundamentally reliability-constrained, even in the absence of explicit sensing--communication resource partitioning.
%%%%%%%%%%%%%%%%%%%%%%%%%%%%%%%%%%%%%%%%%%%%%%%%%%%%%%%%%%%%%%%%%%%%%%%%%%%%%%%%%
%%%%%%%%%%%%%%%%%%%%%%%%%%%%%%%%%%%%%%%%%%%%%%%%%%%%%%%%%%%%%%%%%%%%%%%%%%%%%%%%%
\section{Numerical Results and Comparative Insights}
\label{sec7}

In this section, we provide numerical simulations to validate the sensing--communication tradeoff boundary  derived and evaluate the sensitivity of the proposed framework to blocklength, SNR, and decoding reliability. Furthermore, we contrast the proposed reliability-aware ISAC model to classical resource-allocation-based formulations and identify the distinct operating regimes of FBL decoding uncertainty. 

%=========================================================
\subsection{Simulation Parameters}

The system operates at a carrier frequency of 28 GHz with a bandwidth of 100 MHz. The channel dispersion, target reflection coefficient, thermal noise floor, and maximum transmit power are set to
\[
V=0.5,
\qquad
\xi=1,
\qquad
N_0=0.1,
\qquad
P_{\max}=1,
\]
respectively, unless otherwise specified.

The communication reliability is modeled using the FBL approximation. Unless otherwise specified, the simulations use
\[
n=500,
\qquad
\epsilon=10^{-2}.
\]

The sensing performance metric is defined through the Fisher Information in \eqref{eq12} with $\kappa=1.$

Figures~\ref{fig:pareto_tradeoff}--\ref{fig:mse_crb} consider an AWGN sensing channel having a deterministic channel coefficient $h=1.$

Figures involving multi-SNR analysis use
\[
\mathrm{SNR}\in\{0,5,10,15\}\,\mathrm{dB},
\]
while the multi-blocklength analysis uses
\[
n\in\{100,300,500,1000\}.
\]

Figure~\ref{fig:opt_sensing_rmin} consider block-fading simulations under a Rayleigh fading channel model with
\[
h\sim\mathcal{CN}(0,1),
\]
where the channel envelope remains constant within each communication block and changes independently across blocks. Unless otherwise specified, the block-fading simulations use
$
n=500
$
and are averaged over
$
3\times10^5
$
independent fading realizations per operating point.  For controlled validation purposes, the  decoding uncertainty is modeled through a power-dependent distortion variance of
\[
\sigma_e^2(P)
=
\frac{\alpha}
{1+\exp\left(\beta(P-P_0)\right)},
\]
associated with
$$
\alpha=0.22,
\qquad
\beta=35,
\qquad
P_0=0.22.
$$

This model emulates the rapid improvement in decoding reliability beyond a practical operating threshold under FBL communication.

%=========================================================
\subsection{Impact of Blocklength and Reliability Transition}

\begin{figure}[htbp]
    \centering
    \includegraphics[width=0.95\linewidth]{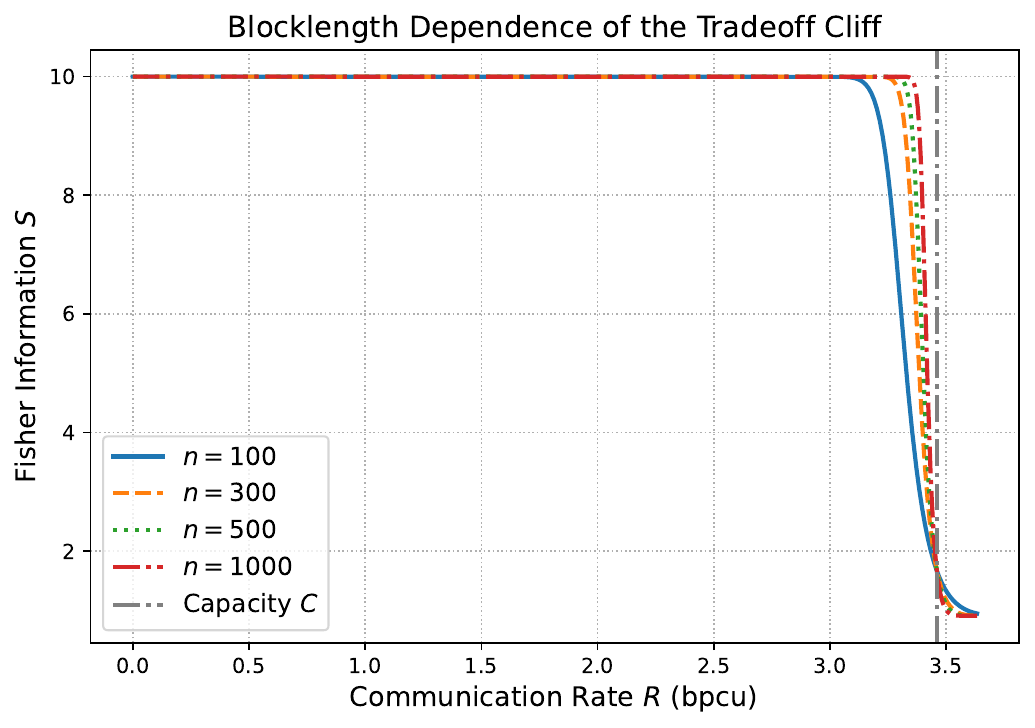}
   \caption{\textbf{Tradeoff Cliff Under Finite Blocklength Communication}: Fisher Information $S$ versus communication rate $R$ for varying blocklengths $n$. The vertical dotted line denotes the Shannon capacity $C$.}
    \label{fig:results_n}
\end{figure}

Fig.~\ref{fig:results_n} illustrates the Fisher Information $S$ as a function of communication rate $R$ for $n \in \{100,300,500,1000\}.$

The results reveal that the communication blocklength fundamentally controls the hardness of the sensing--communication tradeoff. For short packets, the finite blocklength reliability transition is gradual, causing decoding uncertainty to emerge significantly below capacity. Consequently, sensing degradation begins earlier and persists across a wider range of communication rates.

As the blocklength increases, the Tradeoff Cliff becomes progressively sharper and shifts closer to the Shannon capacity. This behavior confirms that high blocklength communication asymptotically meets the classical ISAC assumption of a perfectly known sensing reference. In particular, as
\[
n\rightarrow\infty,
\]
the decoding error probability approaches a step function and the sensing performance remains near the ideal limit of
\[
S\approx \frac{P}{N_0}
\]
for all
\[
R<C.
\]

Furthermore, Theorem~\ref{thm4} formally establishes that the cliff sharpness scales as
\[
\mathcal{O}(\sqrt{n}),
\]
which explains the increasingly abrupt transition observed for high blocklengths.

%=========================================================
\subsection{Comparison with Classical ISAC Tradeoffs}

Traditional ISAC formulations typically interpret the sensing--communication tradeoffs as a resource-allocation problem \textcolor{black}{\cite{Hassanien2016_RadarComm,Liu2021_Resource,Xiong2023_Fundamental}}. In classical power-sharing models, a splitting factor of
\[
\alpha\in[0,1]
\]
shares the transmit power between communication and sensing:
\begin{equation}
R
=
\log_2\left(
1+\frac{\alpha P}{N_0}
\right),
\qquad
S
\propto
\frac{(1-\alpha)P}{N_0}.
\end{equation}

Such formulations imply a relatively smooth and approximately linear tradeoff between communication throughput and sensing quality.

By contrast, the proposed framework reveals an intrinsic reliability-driven coupling even when both functions share the same waveform and the full transmit power budget. Here, sensing degradation emerges not from explicit resource partitioning, but from the uncertainty of the decoded communication waveform itself. As the communication rate approaches the FBL reliability boundary, the waveform becomes an increasingly unreliable sensing reference, producing the sharp Tradeoff Cliff observed in the numerical results.

This distinction highlights a fundamental departure from classical ISAC design: increasing the transmit power alone may not improve sensing performance if the communication rate is simultaneously pushed too close to capacity.

%=========================================================
\subsection{SNR Sensitivity and Capacity Alignment}

\begin{figure}[htbp]
    \centering
    \includegraphics[width=0.95\linewidth]{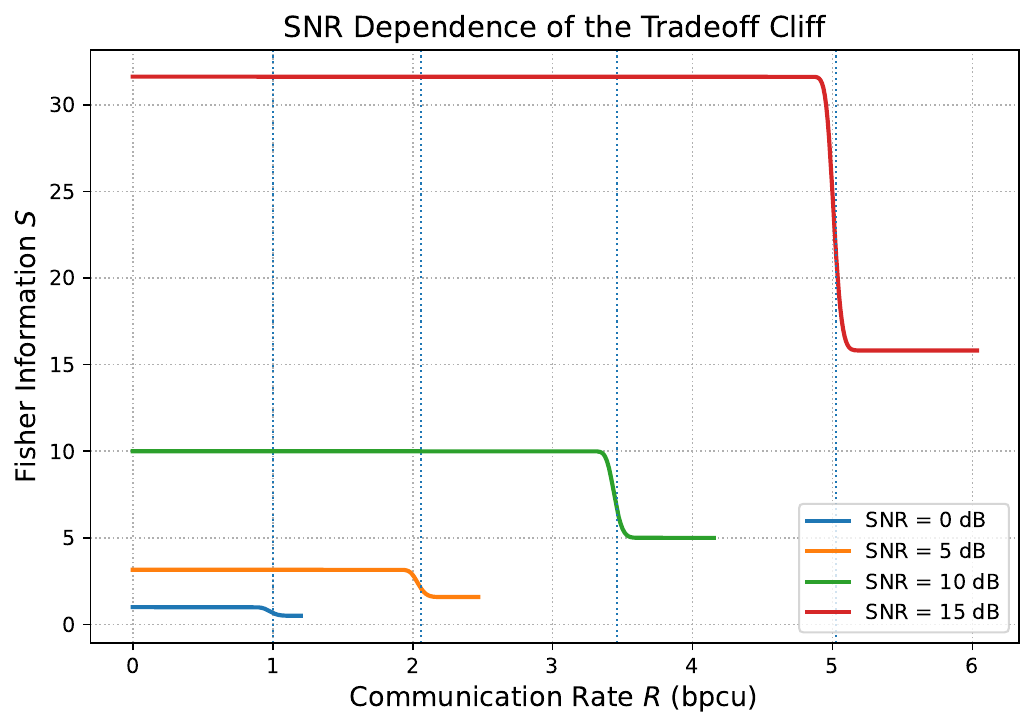}
   \caption{\textbf{SNR Effect}: Fisher Information $S$ versus communication rate $R$ for varying SNR values. The vertical dotted lines denote the corresponding Shannon capacities.}
    \label{fig:snr_sens}
\end{figure}

Fig.~\ref{fig:snr_sens} illustrates the Fisher Information across multiple SNR values.

As expected, increasing the SNR improves the maximum achievable sensing performance through the scaling
\[
S\propto \frac{P}{N_0}.
\]

However, the sensing-optimal communication region remains fundamentally constrained by the Shannon capacity. Although a higher SNR shifts both the capacity and the Tradeoff Cliff toward higher communication rates, the qualitative cliff behavior remains unchanged.

This demonstrates that the Tradeoff Cliff is fundamentally a reliability-driven phenomenon, rather than a purely power-limited effect. Even high-power sensing systems experience rapid performance collapse once the communication rate approaches the FBL reliability boundary.

%=========================================================
\subsection{Validation Under Controlled Covariance Inflation}\label{sec:validation_E}
This experiment isolates the covariance-inflation mechanism
independently of any specific coding architecture, allowing
controlled validation of the analytical Fisher Information model.
Section VI-I subsequently validates that the same behavior
emerges under practical FBL Polar decoding.

 The transmitted waveform is observed through a complex AWGN channel along with additive decoding-induced distortion. The residual decoding uncertainty is modeled as a power-dependent effective distortion variance
\[
\sigma_e^2(P)
=
\frac{\alpha}{1+\exp\left[\beta(P-P_0)\right]},
\]
which captures the rapid reduction in decoding uncertainty beyond a specific reliability transition threshold.

Unless otherwise specified, the simulations use
\[
N_0=0.1,
h=1,
\alpha=0.22,
\beta=35,
P_0=0.22,
\]
with $
3\times10^5$
Monte Carlo realizations for each transmit-power operating point.

The analytical CRB is formulated  as
\[
\mathrm{CRB}(P)
=
\frac{N_0+\sigma_e^2(P)}{P},
\]
while the empirical estimation MSE is obtained by averaging the squared channel-estimation error across all realizations.

\begin{figure}[htbp]
    \centering
    \includegraphics[width=0.95\linewidth]{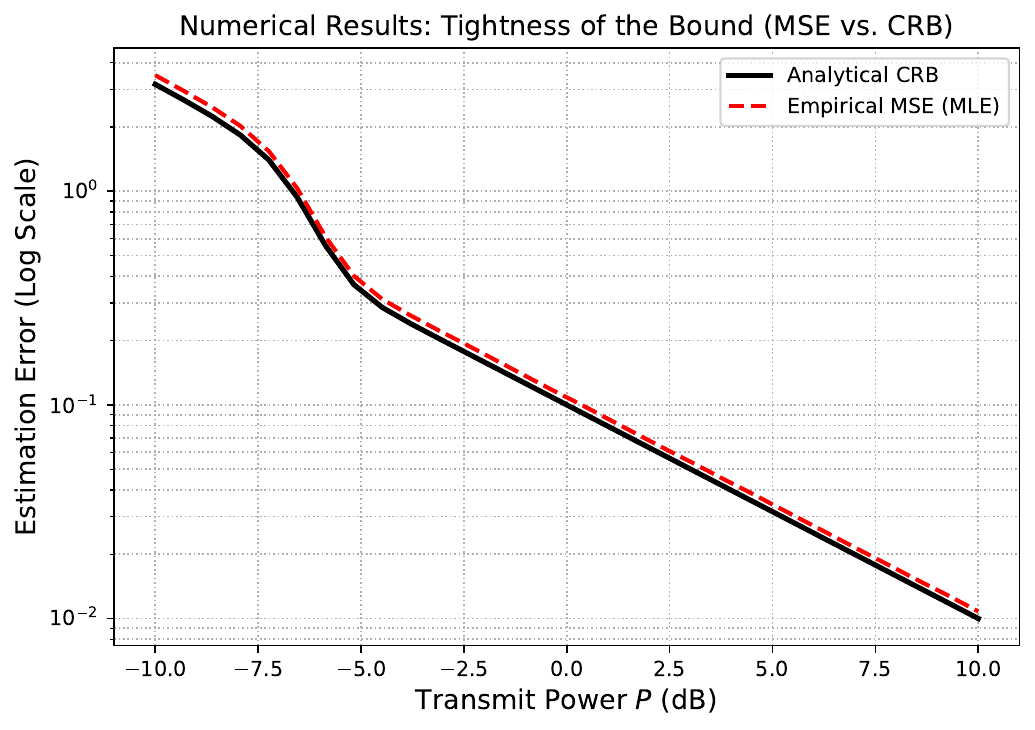}
   \caption{\textbf{CRB Validation}: Analytical CRB and empirical estimation MSE versus transmit power under a reliability-transition decoding uncertainty model. The decoding-induced distortion variance follows a sigmoid transition to emulate the rapid improvement in decoding reliability beyond a practical operating threshold. Results are averaged over \(3\times10^5\) Monte Carlo realizations.}
    \label{fig:mse_crb}
\end{figure}

Fig.~\ref{fig:mse_crb} compares the empirical estimation MSE to  the analytical CRB as a function of transmit power under decoding uncertainty.

At low transmit-power levels, the estimation error is dominated by decoding-induced distortion, resulting in a relatively slow decrease in MSE. As the transmit power increases beyond the reliability-transition threshold, the decoding uncertainty rapidly diminishes, producing the sharp transition region observed near
\[
P\approx2\times10^{-1}.
\]

Beyond this regime, both the empirical MSE and the analytical CRB exhibit the expected inverse-power scaling behavior and closely track each other over the entire operating range. The strong agreement between the analytical and empirical curves validates the proposed Fisher Information characterization and demonstrates that the effective decoding-noise framework accurately captures the impact of communication reliability on sensing performance.

%=========================================================

%=========================================================
\subsection{Impact of QoS Constraints on Sensing Precision}\label{sec:validation_G}

\begin{figure}[htbp]
    \centering
    \includegraphics[width=0.95\linewidth]{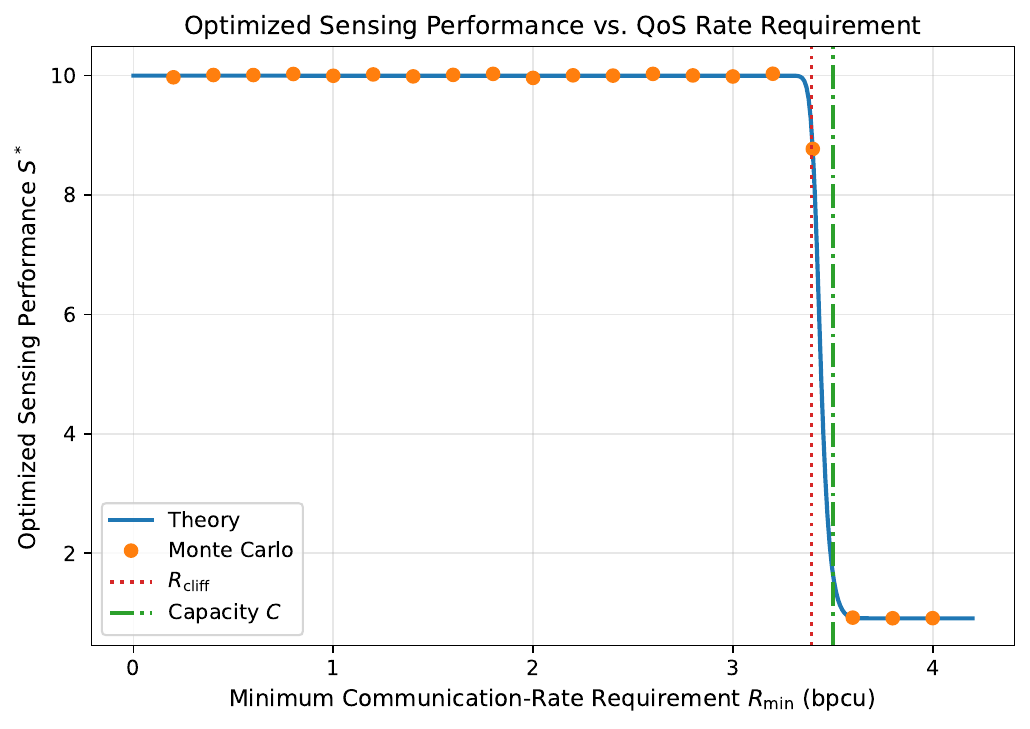}
    \caption{Maximum achievable Fisher Information $S^*$ under varying communication QoS constraints $R_{\min}$. The curve identifies the fundamental boundary of the reliability-aware ISAC feasible region.}
    \label{fig:opt_sensing_rmin}
\end{figure}

Fig.~\ref{fig:opt_sensing_rmin} illustrates the optimized sensing performance
$
S^*$
as a function of the minimum communication-rate requirement
$
R_{\min}.
$
Observe that the optimized sensing performance $S^*$
 closely follows the Fisher Information curve $S(R,n)$ as the minimum communication-rate constraint $R_{\min}$ varies, consistent with the theoretical characterization established in Corollary~\ref{cor_boundary}.

In the low-rate regime, the sensing performance remains near the ideal limit because the system maintains a sufficiently large reliability margin relative to capacity. However, as the QoS requirement approaches the FBL reliability boundary, the sensing performance undergoes a rapid collapse due to sharply increasing decoding uncertainty.
%=========================================================
\subsection{Summary of Operating Regimes}

Based on the numerical analysis, three distinct operating regimes emerge for reliability-aware ISAC systems. These regimes are summarized in Table~\ref{tab:regimes} and provide practical guidelines for balancing communication throughput and sensing fidelity.

\begin{table*}[htbp]
\centering
\caption{Classification of Reliability-Aware ISAC Operating Regimes}
\label{tab:regimes}
\begin{tabular}{@{}lp{3.5cm}p{4.5cm}@{}}
\toprule
\textbf{Regime} & \textbf{Condition} & \textbf{Performance Characterization} \\ \midrule

\textbf{Reliable} 
& 
$R_{\min}\ll C-\Delta_{n}$ 
& 
Near-ideal sensing operation with negligible decoding uncertainty. The sensing performance approaches the classical ISAC limit. 
\\ \addlinespace

\textbf{Transition} 
& 
$R_{\min}\approx C-\Delta_{n}$ 
& 
The Tradeoff Cliff emerges and sensing performance becomes highly sensitive to communication-rate variations. Small rate increases produce disproportionately large sensing degradation.
\\ \addlinespace

\textbf{Saturation} 
& 
$R_{\min}\to C$ 
& 
The communication waveform becomes an unreliable sensing reference due to high decoding uncertainty, causing sensing-information collapse.
\\ \bottomrule

\end{tabular}

\vspace{2pt}

\begin{flushleft}
\footnotesize{
Note: $\Delta_{n}$ denotes the finite blocklength reliability margin determined by the communication blocklength \(n\).
}
\end{flushleft}
\end{table*}
%%%%%%%%%%%%%%%%%%%%%%%%%%%%%%%%%%%%%%%%%%%%%%%%%%%%%%%%%%%%%%%%%%%
%-----------------------------------------------
\subsection{Polar-Coded Finite Blocklength Validation}
\label{sec:polar_validation}
To validate the covariance-based framework, we perform a practical finite-blocklength (FBL) simulation using Polar-coded short-packet communication \cite{Arikan2009, TalVardy2015,condo2021short}. In contrast to theoretical models, this experiment explicitly generates residuals $\mathbf{e} = \hat{\mathbf{x}}-\mathbf{x}$ through practical Polar encoding/decoding \textcolor{black}{over two channel models: an AWGN channel with BPSK modulation, and a Rayleigh-fading channel with coherent QPSK and receiver channel-state information (CSIR), the latter using a fixed, randomly-generated bit interleaver that spreads each codeword's coded bits across the in-phase and quadrature branches of the QPSK symbol stream so that a single deep fade does not correlate the failure of adjacent polar bit-channel positions. Both channels use $n=256$ and a rate sweep $R \in [0.1, 1.2C]$; the AWGN and Rayleigh operating SNRs (5~dB and 12~dB respectively) are chosen independently so that the two channels present comparable capacity and dispersion, allowing their finite-blocklength transitions to be compared on a like-for-like basis.}

\textcolor{black}{The AWGN construction uses the standard Gaussian-Approximation (GA) polar construction. For Rayleigh, the frozen-bit set is instead selected using a channel-matched construction: a Monte-Carlo density-evolution (MC-DE) procedure that recursively synthesizes each bit-channel's LLR distribution under the true Rayleigh/QPSK statistics, with the resulting per-subchannel reliability estimate averaged over several independent realizations to reduce Monte-Carlo estimation variance.}

\textcolor{black}{Note that this framework does not require that the residual $\mathbf{e}$ is Gaussian-distributed at the sample level. For discrete-modulation symbols, $\mathbf{e}$ is inherently discrete-valued, and its per-block distribution is a mixture of a dominant near-zero component from correctly-decoded blocks and a heavier-tailed component from decoding failures. What the framework requires is the weaker, second-moment condition that $\hat\sigma_e^2$ serves as a valid covariance-equivalent parameter for the linearized Fisher-Information analysis of Section~III, analogous to standard effective-noise-covariance treatments used elsewhere in estimation theory when the exact error distribution is analytically intractable. }

\textcolor{black}{For each rate point, the empirical block error rate (BLER) and residual covariance $\hat{\sigma}_e^2 = \frac{1}{n}\|\mathbf{e}\|^2$ are obtained from an adaptive hybrid estimator rather than a fixed number of Monte Carlo trials. A short pilot run at each rate first estimates the order of magnitude of the block error rate; direct Monte Carlo is used whenever the resulting precision target is reachable within a bounded number of direct trials, and a defensive-mixture importance-sampling estimator -- a biased noise/deep-fade proposal distribution mixed with the true channel law, with importance weights bounded by construction -- is used for rate points where the pilot indicates a genuinely rare event beyond that direct-simulation budget. We verified this estimator against direct Monte Carlo at several rates where both are affordable, finding agreement well within one combined standard error at every point checked. Rate points whose calibrated theoretical BLER falls below $10^{-13}$ are excluded from simulation entirely.}

\textcolor{black}{We additionally evaluate a minimum-distance-aware flip decoder (DFD) \cite{dfd1,dfd2} alongside plain successive-cancellation (SC) decoding, applied to the Rayleigh channel only; this is included to confirm that the empirically observed finite-blocklength transition and Tradeoff Cliff geometry are not an artifact specific to SC decoding, rather than as a decoder-comparison study in its own right.}

\textcolor{black}{The covariance-equivalent finite-blocklength reference is instantiated by estimating the system-dependent distortion factor \(\kappa\) from the resolved empirical decoder residuals, consistent with \(\sigma_e^2(R,n)=\kappa P_e(R,n)\).}

\begin{figure*}[htbp]
    \centering
    \subfloat[Empirical vs. calibrated-theoretical residual covariance inflation\textcolor{black}{, AWGN and Rayleigh, SC (and Rayleigh DFD) decoding}.]{\includegraphics[width=0.47\linewidth]{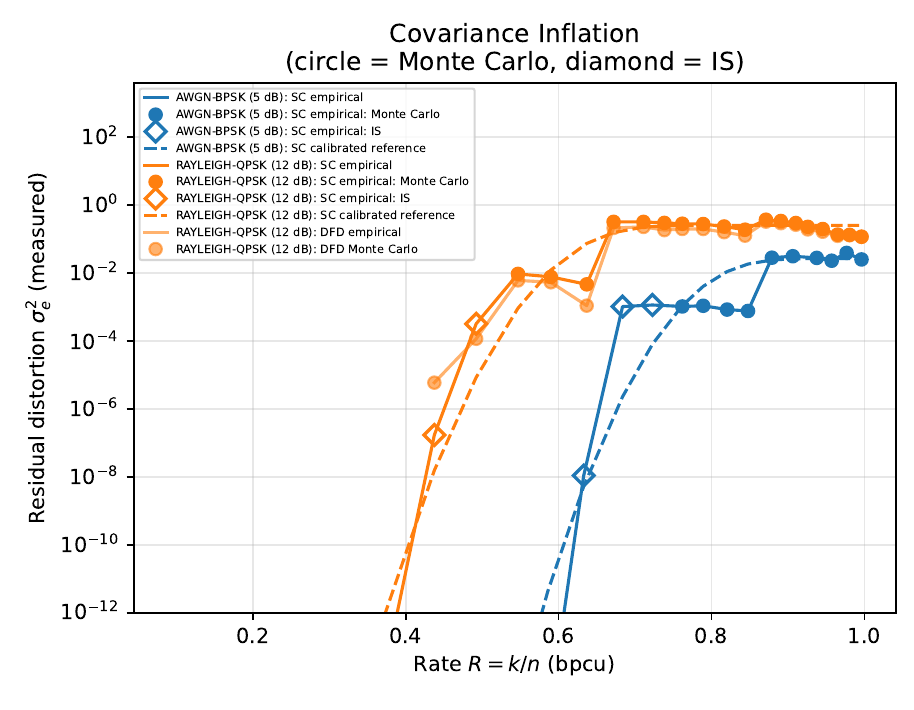}\label{fig:codeaware_covariance}}
    \hfill
    \subfloat[Polar-coded sensing FI vs. rate (Tradeoff Cliff)\textcolor{black}{, both channels and decoders}.]{\includegraphics[width=0.47\linewidth]{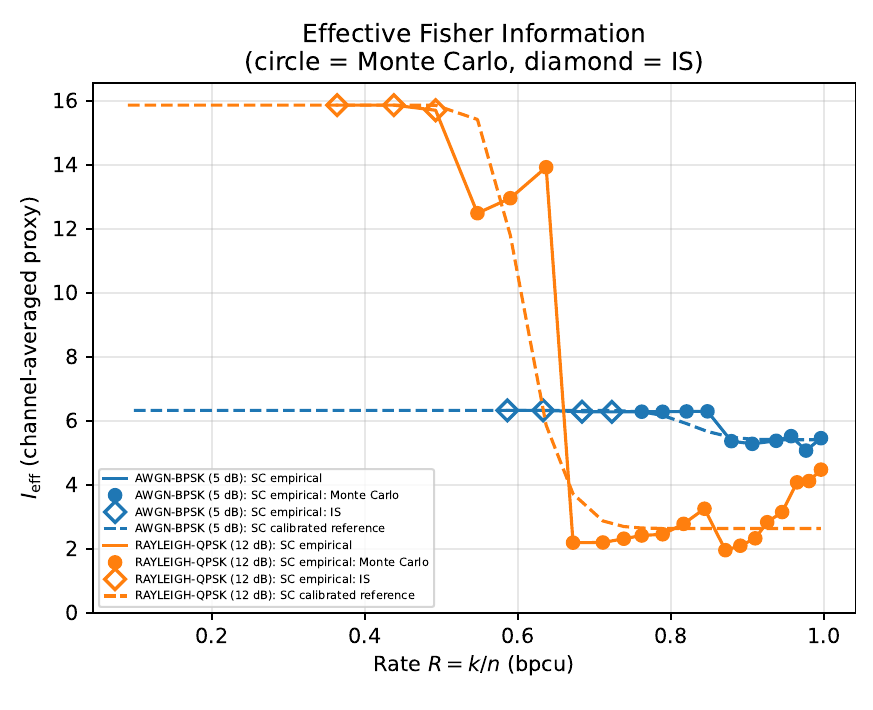}\label{fig:codeaware_tradeoff}}
    \\
    \subfloat[Polar decoder BLER vs. \textcolor{black}{calibrated} finite-blocklength \textcolor{black}{reference, both channels and decoders}.]{\includegraphics[width=0.47\linewidth]{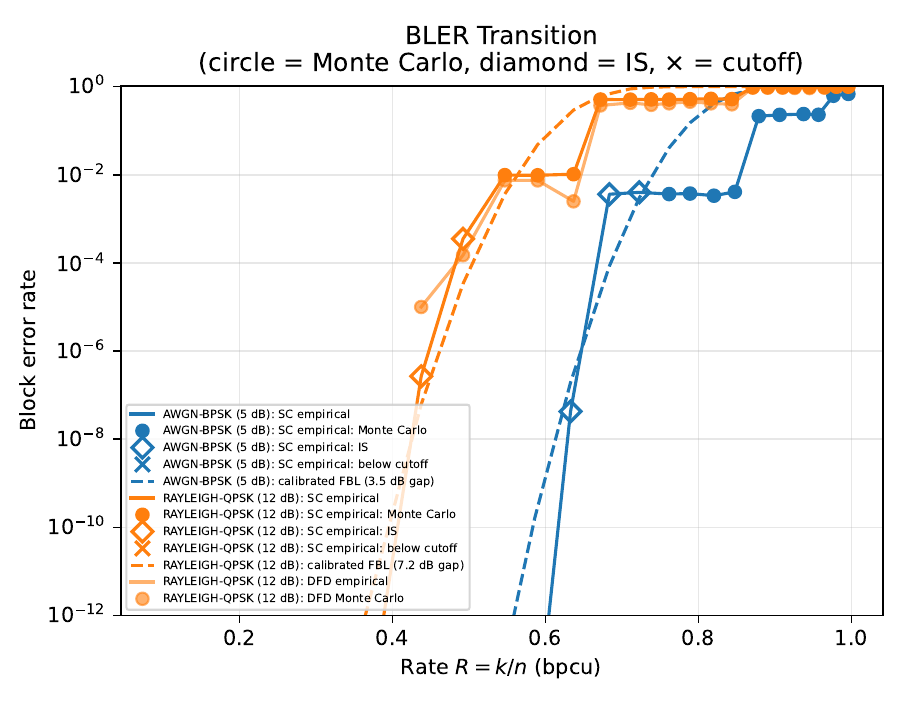}\label{fig:bler_transition}}
     \hfill
     \subfloat[CRB vs. Rate for the Rayleigh/QPSK configuration at 0, 6, 12, and 18~dB, empirical (markers) vs. calibrated theory (dashed) per SNR.]{\includegraphics[width=0.47\linewidth,height=2.7in]{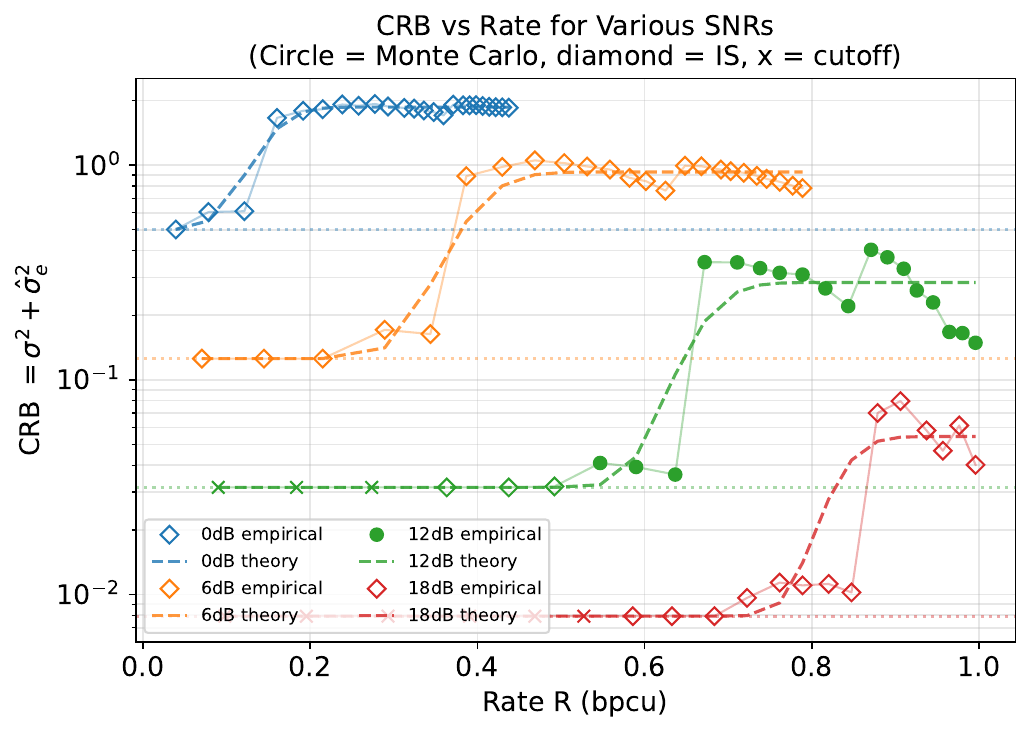}\label{fig:crb_multisnr}}
    \caption{Practical validation of the reliability-driven sensing framework using Polar-coded communication\textcolor{black}{, across AWGN/BPSK and Rayleigh/QPSK channels and SC/DFD decoding}.}
    \label{fig:polar_validation_results}
\end{figure*}

\subsubsection{Results and Discussion}
\textcolor{black}{Fig.~\ref{fig:codeaware_covariance} confirms that the Polar-decoded residual covariance remains negligible in the reliable regime, but increases sharply near the FBL boundary, following the predicted $\mathbf{R}_e \approx P_e \mathbf{R}_{\mathrm{fail}}$ relation from Appendix~\ref{app:model}, consistently across both channel models and decoding algorithms tested. The fitted covariance-inflation constant is $\kappa \approx 0.027$ for AWGN (SC), and $\kappa \approx 0.25$ (SC) / $0.20$ (DFD) for Rayleigh -- as noted above, this constant is setting-dependent rather than universal. Consequently, the sensing FI in Fig.~\ref{fig:codeaware_tradeoff} matches the analytical Tradeoff Cliff geometry derived in Section~\ref{sec4}.  }

\textcolor{black}{Fig.~\ref{fig:bler_transition} shows the empirical Polar BLER transition against its calibrated finite-blocklength reference. Rather than a single smooth transition, the empirical curve exhibits a stepwise (plateau-then-jump) structure: BLER stays approximately constant over a contiguous range of rates and then rises sharply at specific rate values, repeating this pattern across the sweep. This is explained by the code's own minimum distance structure rather than by measurement noise, and holds regardless of which reliability ordering is used to select the frozen set: for the (non-bit-reversed) Arikan transform used here, the row weight of subchannel $i$ is exactly $2^{\mathrm{popcount}(i)}$, so the code's minimum distance is $d_{\min} = 2^{\min_i \mathrm{popcount}(i)}$, the minimum taken over the information set actually selected. As \(k\) increases, the largest feasible minimum-row-weight tier remains unchanged over a range of rates and then drops discretely when that tier can no longer accommodate \(k\) information positions. The curve steps upward only once $k$ increases enough to require a bit-channel from a strictly lower popcount tier, dropping $d_{\min}$ discretely. This produces the plateau-then-jump structure visible in Fig.~\ref{fig:bler_transition}, a feature of a fixed, practically-constructed polar code under SC-family decoding, distinct from the smooth transition of finite-blocklength achievability bound, which characterizes the best possible code at each blocklength and rate.} These results demonstrate that the proposed model accurately captures the physical degradation emerging naturally from practical Polar decoding behavior, rather than being a mere artifact of the Gaussian abstraction. \textcolor{black}{We extend the SC-decoder validation to  characterizing how the Tradeoff Cliff shifts with SNR. Fig.~\ref{fig:crb_multisnr} shows the empirical CRB $= \sigma^2 + \hat\sigma_e^2$ against rate for the Rayleigh/QPSK configuration at four operating SNRs (0, 6, 12, and 18~dB), using the same real polar encoder/SC-decoder pipeline, adaptive direct/importance-sampling estimator, and SNR-gap calibration described above. The fitted covariance-inflation constant $\kappa$ decreases monotonically with SNR (1.367, 0.804, 0.253, and 0.047 at 0, 6, 12, and 18~dB respectively), consistent with the setting-dependence of $\kappa$ already noted, and in qualitative agreement with the analytical scaling of Section~IV.}

\textcolor{black}{ \textbf{We emphasize that this validation is deliberately
constructed without assuming the Gaussian covariance-equivalent model:
the residuals $e = \hat{x}-x$ and the resulting empirical covariance
$\hat{\sigma}_e^2$ in Fig.~\ref{fig:polar_validation_results} are measured directly from a
practical Polar encoder/decoder, providing independent empirical
support for the covariance-error-probability relation
$\sigma_e^2 \approx \kappa P_e$ and for the Gaussian-equivalent Fisher
Information framework of Section~III, beyond the synthetic
covariance-inflation experiments of Sections~VI.\ref{sec:validation_E}--\ref{sec:validation_G}.}} \textcolor{black}{A broader empirical characterization spanning additional blocklengths, modulation formats, and code families such as LDPC, together with an independently-calibrated DFD reference curve, remains an important direction for future work (Section~VII.C).}
%%%%%%%%%%%%%%%%%%%%%%%%%%%%%%%%%%%%%%%%%%%%%%%%%%%%%%%%%%%%%%%%%%%%%%%%%
\section{Conclusion and Future Work}\label{sec8}
The fundamental limits of Integrated Sensing and Communications (ISAC) \textcolor{black}{within the adopted decoded-reference based sensing framework was} \textbf{explored } by relaxing the classical assumption of a perfect sensing reference waveform. By incorporating FBL theory, we revealed the ``Tradeoff Cliff''---a sharp, nonlinear degradation in sensing performance that occurs as the communication rate approaches Shannon capacity.

\subsection{Summary of Contributions}
Our analysis yields several key insights into reliability-aware ISAC:
\begin{itemize}
    \item \textbf{Reliability-Driven Coupling:} Fisher Information is intrinsically tied to decoding success. Consequently, a unified waveform's utility as a sensing reference fundamentally depends on communication reliability.
    \item \textbf{The Tradeoff Cliff:} While sensing is near-optimal at low rates, it degrades precipitously as the reliability margin vanishes ($R \to C$).
    \item \textbf{Blocklength Scaling:} The blocklength $n$ governs the tradeoff geometry. Short packets in URLLC induce an earlier, gradual degradation, whereas long packets (eMBB) delay the collapse but sharpen the cliff.
    \item \textbf{Optimal Design:} We established an optimization framework that maximizes the achievable sensing precision under strict communication QoS constraints, offering a practical tool for 6G resource allocation.
\end{itemize}

\textbf{Limitations:} By modeling the decoding uncertainty as an equivalent second-order Gaussian distortion, our framework captures aggregate sensing degradation. It does not explicitly account for code-specific structured failures or burst errors, making the extension to code-aware sensing analysis an important next step. 

\subsection{Implications for 6G}
Our findings demonstrate that for unified ISAC waveforms and FBL decoded-reference based ISAC systems,
operating arbitrarily close to the communication reliability
boundary become fundamentally incompatible with maintaining
high-precision sensing performance. System designers must enforce an information-theoretic \textit{reliability margin} to ensure that the waveform remains a stable sensing reference, moving beyond purely hardware-centric design constraints.

\subsection{Future Research Directions}
Building on this framework, future work could investigate:
\begin{itemize}
    \item \textbf{Exact Analysis:} The analytical framework relies on a covariance-equivalent approximation of decoding uncertainty and on the corresponding tractable Fisher Information characterization. Extending the analysis to exact code-aware residual statistics remains an important direction for future work.
    \item \textbf{MIMO ISAC Tradeoffs:} Extending the reliability-driven model to multi-antenna systems to explore how spatial degrees of freedom might mitigate decoding uncertainty and to other extensions of ISAC like integrated polarimetric  sensing and communications (IPSAC)\cite{IPSAC}.
    \item \textbf{Imperfect CSI:} Analyzing how channel estimation errors compound the instability of the sensing reference.
    \item \textbf{Error-Resilient Coding:} Designing specialized codes that minimize sensing variance even in the event of communication decoding failures
    \textcolor{black}{\item \textbf{Experimental Validation} Experimental validation of the covariance-equivalent model across different coding schemes, modulation formats, channel models, and hardware platforms constitutes an important direction for future research.}
  
\end{itemize}
%%%%%%%%%%%%%%%%%%%%%%%%%%%%%%%%%%%%%%%%%%%%%%%%%%%%%%%%%%%%%%%%%%%%%%%%%%%%%%
%%%%%%%%%%%%%%%%%%%%%%%%%%%%%%%%%%%%%%%%%%%%%%%%%%%%%%%%%%%%%%%%%%%%%%%%%%%%%%%%%%
%%%%%%%%%%%%%%%%%%%%%%%%%%%%%%%%%%%%%%%%%%%%%%%%%%%%%%%%%%%%%%%%%%%%%%%%%%%%%%%%%%%%
\appendices
\section{Covariance Interpretation of Decoding Uncertainty}
\label{app:model}

This appendix provides the statistical foundation for the equivalent decoding-distortion model used for characterizing sensing degradation under FBL uncertainty.

\subsection{Decoder Residual and Covariance Decomposition}
Let $\mathbf{x} \in \mathbb{C}^{n}$ be a coded block and $\hat{\mathbf{x}} = \mathcal{D}(\mathbf{y})$ be the output of a FBL decoder. We define the decoding residual as $\mathbf{e} = \hat{\mathbf{x}} - \mathbf{x}$, such that the decoded waveform is $\hat{\mathbf{x}} = \mathbf{x} + \mathbf{e}$. 

While the exact structure of $\mathbf{e}$ depends on the specific code and SNR, we adopt a \emph{covariance-equivalent abstraction} to capture its second-order impact on sensing. Let $\mathcal{E} = \{\hat{\mathbf{x}} \neq \mathbf{x}\}$ be the decoding error event with probability $P_e$. By the law of total expectation, the residual covariance $\mathbf{R}_e = \mathbb{E}[\mathbf{e}\mathbf{e}^H]$ decomposes as:
\begin{equation}
\mathbf{R}_e = (1-P_e)\mathbf{R}_{\mathrm{succ}} + P_e\mathbf{R}_{\mathrm{fail}},
\label{eq:mixture_covariance}
\end{equation}
where $\mathbf{R}_{\mathrm{succ}}$ and $\mathbf{R}_{\mathrm{fail}}$ are the residual covariances conditioned on success and failure, respectively. In reliable FBL systems, the distortion under successful decoding is negligible ($\mathbf{R}_{\mathrm{succ}} \approx \mathbf{0}$) \cite{Polyanskiy2010_FBL}, implying that the aggregate uncertainty scales directly with the error probability:
\begin{equation}
\mathbf{R}_e \approx P_e \mathbf{R}_{\mathrm{fail}}.
\end{equation}

\subsection{Isotropic Approximation and Gaussian Interpretation}
In interleaved coded systems, Fisher Information is primarily affected by the aggregate covariance magnitude, rather than by fine symbol-level correlations. Assuming no directional preference in the residual, we model the failure covariance as isotropic: $\mathbf{R}_{\mathrm{fail}} \approx \kappa \mathbf{I}$, where $\kappa > 0$ is a system-dependent scaling parameter. This yields the model:
\begin{equation}
\mathbf{R}_e \approx \kappa P_e \mathbf{I} \implies e_k \sim \mathcal{CN}(0, \sigma_e^2), \quad \sigma_e^2 = \kappa P_e.
\end{equation}

\textcolor{black}{A Gaussian distribution is adopted as a covariance-equivalent analytical surrogate; Gaussianity of the decoder residual itself is not required. In systems containing multiple approximately independent impairment sources, a Gaussian aggregate may additionally be motivated asymptotically by central-limit arguments.} If $e_k$ arises from multiple independent perturbation sources (e.g., quantization, synchronization mismatch, and interference), the aggregate residual converges in distribution to $\mathcal{CN}(0, \sigma_e^2)$. This allows for an analytically tractable Fisher Information analysis that still captures the dominant physical degradation.

\subsection{Scope and Limitations}
This formulation is an equivalent sensing-distortion abstraction, rather than an exact physical decoder model. \textcolor{black}{Accordingly, the Gaussian decoding-distortion
model should be interpreted as a covariance-equivalent representation of
the aggregate decoding uncertainty, rather than as the exact statistical
distribution of decoder residuals.} It is most accurate for metrics dependent on second-order statistics (FI and CRB) and does not explicitly capture:
\begin{itemize}
    \item Code-specific structured failures or burst-error behavior.
    \item Symbol-dependent correlations or non-Gaussian residual distributions.
\end{itemize}

\textcolor{black}{The proposed approximation is therefore most
appropriate for the analytical characterization of sensing performance in
decoded-reference ISAC systems, where the dominant effect of decoding
uncertainty is captured through its covariance.} 

Despite these simplifications, the framework provides a mathematically rigorous mechanism for connecting communication reliability with Fisher Information loss. \textcolor{black}{\textbf{An empirical instance of this connection under
a real, code-specific decoder --- rather than the idealized Gaussian
residual model assumed above --- is validated in
Section~VI.\ref{sec:polar_validation} (Fig.~\ref{fig:polar_validation_results}).}}
%%%%%%%%%%%%%%%%%%%%%%%%%%%%%%%%%%%%%%%%%%%%%%%%%%%%%%%%%%%%%%%%%%%%%%%%%%%%%%%%%%%%
%%%%%%%%%%%%%%%%%%%%%%%%%%%%%%%%%%%%%%%%%%%%%%%%%%%%%%%%%%%%%%%%%%%%%%%%%%%%%%%%%%%%
\section{Proof of Theorem 1}
\label{app:FI_proof}

This appendix derives the Fisher Information  for the complex scalar channel gain $h$ under conditions of decoding uncertainty. \textcolor{black}{
The expectation in the following derivation is taken jointly over the
observation noise and the decoded-reference uncertainty. Equivalently,
the Fisher Information derived here corresponds to the expected Fisher
Information
\[
I(h)=\mathbb{E}_{\hat{\mathbf{x}}}
\!\left[
I(h|\hat{\mathbf{x}})
\right],
\]
where the inner Fisher Information is conditioned on a particular
decoded waveform realization.}
\subsection{Observation Model}
The received sensing signal is $y = hx + w$, where $w \sim \mathcal{CN}(0,N_0)$. The receiver only has access to the decoded waveform $\hat{x} = x + e$, with $e \sim \mathcal{CN}(0,\sigma_e^2)$ being  independent of $x$ and $w$. Substituting $x = \hat{x} - e$ into the observation model yields:
\begin{equation}
y = h\hat{x} + \tilde{w}, \quad \tilde{w} = w - he.
\end{equation}
Since $w$ and $e$ are independent Gaussian variables, the effective noise $\tilde{w}$ is distributed as $\mathcal{CN}(0, \sigma_{\tilde{w}}^2)$ with variance $\sigma_{\tilde{w}}^2 = N_0 + |h|^2\sigma_e^2$. 
Under the covariance-equivalent approximation developed in
Appendix~\ref{app:model}, the observation model is approximated as
\[
y|\hat{x}\approx \mathcal{CN}(h\hat{x},\sigma_{\tilde w}^2),
\]
where the approximation neglects higher-order residual
correlation terms induced by the coupling between
$\hat{x}$ and the decoding residual.

\textcolor{black}{
\begin{prop}[Accuracy of the Covariance-Equivalent Approximation]
\label{prop:covariance_equivalent}
Consider the decoded-reference based  sensing model $\hat{\mathbf{x}}=\mathbf{x}+\mathbf{e}$ and $\tilde{\mathbf{w}} = \mathbf{w}-\alpha\mathbf{e}$, where $\mathbf{e}\sim\mathcal{CN}(\mathbf{0},\sigma_e^2\mathbf{I})$ and $\mathbf{w}\sim\mathcal{CN}(\mathbf{0},N_0\mathbf{I})$ are independent.
\\
Under moderate decoding uncertainty ($\sigma_e^2\ll\min(P,N_0)$), the exact Fisher Information admits the perturbation expansion of
\\
\begin{equation}
\mathbf{I} = \mathbf{I}_{\rm eq} + \Delta\mathbf{I},
\end{equation}
\\
where $\mathbf{I}_{\rm eq}$ is the Fisher Information obtained via the proposed covariance-equivalent model, and the perturbation from the neglected cross-covariance matrix $\mathbf{\Delta}$ satisfies $\Delta\mathbf{I} = O\left(\left\| \mathbf{\Sigma}_{\rm eq}^{-1} \mathbf{\Delta} \right\|\right)$.
\\
Consequently, the proposed covariance-equivalent Fisher Information is the leading-order approximation to the exact Fisher Information, while the correlation between the decoded waveform and effective noise appears only as a higher-order perturbation.
\end{prop}}
\begin{proof}
\textcolor{black}{
The decoded reference waveform and effective noise have the exact joint
covariance
\[
\mathbf{\Sigma}
=
\begin{bmatrix}
(P+\sigma_e^2)\mathbf{I}
&
-\alpha\textcolor{black}{^*}\sigma_e^2\mathbf{I}
\\
-\alpha\sigma_e^2\mathbf{I}
&
(N_0+|\alpha|^2\sigma_e^2)\mathbf{I}
\end{bmatrix}
=
\mathbf{\Sigma}_{\rm eq}
+
\mathbf{\Delta},
\]
where
\[
\mathbf{\Sigma}_{\rm eq}
=
\begin{bmatrix}
(P+\sigma_e^2)\mathbf{I}
&
\mathbf{0}
\\
\mathbf{0}
&
(N_0+|\alpha|^2\sigma_e^2)\mathbf{I}
\end{bmatrix}.
\]
Since
\[
\sigma_e^2\ll\min(P,N_0),
\]
we have
\[
\|
\mathbf{\Sigma}_{\rm eq}^{-1}
\mathbf{\Delta}
\|
<
1,
\]
and therefore the inverse covariance admits the convergent Neumann
series expansion
\[
\mathbf{\Sigma}^{-1}
=
\mathbf{\Sigma}_{\rm eq}^{-1}
-
\mathbf{\Sigma}_{\rm eq}^{-1}
\mathbf{\Delta}
\mathbf{\Sigma}_{\rm eq}^{-1}
+
O\!\left(
\|
\mathbf{\Sigma}_{\rm eq}^{-1}
\mathbf{\Delta}
\|^2
\right),
\]
following the standard matrix perturbation result
\cite{HornJohnson2013}.
Since the Fisher Information for a complex Gaussian model depends on the
inverse covariance matrix, substitution of the above expansion yields
\[
\mathbf{I}
=
\mathbf{I}_{\rm eq}
+
O\!\left(
\|
\mathbf{\Sigma}_{\rm eq}^{-1}
\mathbf{\Delta}
\|
\right).
\]
Hence, the covariance-equivalent Fisher Information relied upon in this
treatise corresponds to the leading-order term of the exact Fisher
Information, while the effect of the neglected cross-covariance appears
only as a perturbation.}
\end{proof}
%%%%%%%%%%%%%%%%%%%%%%%%%%%%%%%%%%%%%%%%%%%%%%%%%%%%%%%%%%%%%%%%%%%%%%%%%%%%%%%%%%%%%%%%%%%%%%%%%%%%%%%%%%%%%%%%%%%%%%%%%
\textcolor{black}{\subsection{Step 1: Conditional Fisher Information for Fixed $\hat{x}$}
For a \emph{fixed} realization of the decoded reference $\hat{x}$, the
conditional log-likelihood is given by
\begin{equation}
    \log p(y|h,\hat{x}) = -\frac{|y-h\hat{x}|^2}{\sigma_{\tilde{w}}^2}
    - \log(\pi \sigma_{\tilde{w}}^2).
    \label{eq:loglik_cond}
\end{equation}
Note that the effective variance $\sigma_{\tilde{w}}^2 = N_0 +
|h|^2\sigma_e^2$ depends on $h$ through its second moment $\sigma_e^2$,
but \emph{not} on the specific realization of $\hat{x}$ itself, since
$\sigma_e^2$ is a property of the decoding channel rather than of any
particular decoded outcome. With the Wirtinger derivative
$\partial \sigma_{\tilde{w}}^2/\partial h^* = h\sigma_e^2$, the score
function conditioned on $\hat{x}$ is
\begin{equation}
    \frac{\partial}{\partial h^*}\log p(y|h,\hat{x}) =
    \frac{\hat{x}\textcolor{black}{^*}(y-h\hat{x})}{\sigma_{\tilde{w}}^2}
    \textcolor{black}{+} \frac{h\sigma_e^2}{\sigma_{\tilde{w}}^2}
    \left( \frac{|y-h\hat{x}|^2}{\sigma_{\tilde{w}}^2}\textcolor{black}{-}1\right).
    \label{eq:score_cond}
\end{equation}
Defining the residual $\epsilon = y - h\hat{x} = \tilde{w}$, with
$\epsilon \sim \mathcal{CN}(0,\sigma_{\tilde{w}}^2)$, and evaluating
$I(h|\hat{x}) = \mathbb{E}\big[|\partial \log p(y|h,\hat{x})/\partial
h^*|^2 \,\big|\, \hat{x}\big]$ under the equivalent covariance-based
decoding uncertainty model of Appendix A, while retaining terms up to
the moderate-decoding-uncertainty approximation, yields
\begin{equation}
    I(h|\hat{x}) = \frac{|\hat{x}|^2}{N_0+|h|^2\sigma_e^2}
    + \frac{|h|^2\sigma_e^4}{(N_0+|h|^2\sigma_e^2)^2}.
    \label{eq:FI_cond}
\end{equation}
}
\textcolor{black}{\subsection{Step 2: Expectation over $\hat{x}$}
Since $\hat{x} = x + e$ is itself random due to finite-blocklength
decoding, the operationally relevant sensing metric is the Fisher
Information averaged over the distribution of $\hat{x}$,
\begin{equation}
    I(h) = \mathbb{E}_{\hat{x}}\big[I(h|\hat{x})\big].
    \label{eq:FI_avg}
\end{equation}
Inspection of \eqref{eq:FI_cond} shows that $I(h|\hat{x})$ depends on
$\hat{x}$ only through $|\hat{x}|^2$, appearing solely in the
mean-gradient term. The variance-gradient term (the second term of
\eqref{eq:FI_cond}) does not depend on $\hat{x}$ directly, only on
$\sigma_e^2$. Consequently, taking the expectation in
\eqref{eq:FI_avg} reduces to replacing $|\hat{x}|^2$ by its expected
value,
\begin{equation}
    \mathbb{E}[|\hat{x}|^2] = P + \sigma_e^2 \approx P,
    \label{eq:E_xhat}
\end{equation}
which gives
\begin{equation}
    I(h) \approx \frac{\mathbb{E}[|\hat{x}|^2]}{N_0+|h|^2\sigma_e^2}
    + \frac{|h|^2\sigma_e^4}{(N_0+|h|^2\sigma_e^2)^2}.
    \label{eq:FI_final}
\end{equation}
This confirms that the conditional derivation of Step~1 and the
expectation of Step~2 coincide in reducing to the single substitution
\eqref{eq:E_xhat}, since $I(h|\hat{x})$ has no dependence on $\hat{x}$
beyond its second moment. Equation~\eqref{eq:FI_final} is precisely
\eqref{eq8} of Theorem~1.
\\
\begin{rem}
Because $I(h|\hat{x})$ depends on $\hat{x}$ only through $|\hat{x}|^2$,
the conditional-then-averaged construction above and the direct
single-step derivation based on the averaged observation model
$y|\hat{x}\sim\mathcal{CN}(h\hat{x},\sigma_{\tilde{w}}^2)$ yield the same
result. This equivalence relies on the covariance-equivalent modeling
assumption of Appendix A, under which $\sigma_{\tilde{w}}^2$ depends on
$\hat{x}$ only through the second moment - $\sigma_e^2$, rather than
through the specific realized deviation $e = \hat{x}-x$.
\end{rem}
}
%%%%%%%%%%%%%%%%%%%%%%%%%%%%%%%%%%%%%%%%%%%%%%%%%%%%%%%%%%%%%%%%%%%%%%%%%%%%%%%%%%%%%%%%%%%%%%%%%%%%%%%%%%%%%%%%%%%%%%%%%
\subsection{Approximation Justification}
\label{AG}
When the uncertainty-induced noise is small relative to thermal noise ($|h|^2\sigma_e^2 \ll N_0$), the variance-gradient term (the second term in \eqref{eq:FI_final}) becomes negligible. Consequently:
\begin{equation}
I(h) \approx \frac{\mathbb{E}[|\hat{x}|^2]}{N_0 + |h|^2\sigma_e^2}.
\end{equation}
Finally, noting that $\mathbb{E}[|\hat{x}|^2] = P + \sigma_e^2 \approx P$, the Fisher Information simplifies to the equivalent noise formulation:
\begin{equation}
I(h) \approx \frac{P}{N_0 + |h|^2\sigma_e^2}.
\end{equation}
This completes the proof.
%%%%%%%%%%%%%%%%%%%%%%%%%%%%%%%%%%%%%%%%%%%%%%%%%%%%%%%%%%%%%%%%%%%%%%%%%%%%%%%%
%%%%%%%%%%%%%%%%%%%%%%%%%%%%%%%%%%%%%%%%%%%%%%%%%%%%%%%%%%%%%%%%%%%%%%%%%%%%%%%%%
{\color{black}
\section{Proof of Theorem 2: Vector Fisher Information for Multi-Parameter Estimation}
\label{app:vector_fi}}

We extend the scalar Fisher Information analysis to the joint estimation of multiple sensing parameters: namely to delay, Doppler, and complex channel gain.

\subsection{Signal Model and Parameter Vector}
Consider a parametric received signal $y(t) = \alpha x(t - \tau) e^{j2\pi \nu t} + w(t)$, where $\alpha \in \mathbb{C}$ is the complex reflection coefficient, $\tau$ is the delay, $\nu$ is the Doppler shift, and $w(t) \sim \mathcal{CN}(0, N_0)$ is the additive thermal noise. 

Due to decoding uncertainty, the sensing receiver operates on the decoded estimate $\hat{x}(t) = x(t) + e(t)$, with error $e(t) \sim \mathcal{CN}(0, \sigma_e^2)$. Substituting $x(t)$ into the observation yields:
\begin{equation}
y(t) = \alpha \hat{x}(t - \tau) e^{j2\pi \nu t} + \tilde{w}(t),
\end{equation}
where the effective noise $\tilde{w}(t) = w(t) - \alpha e(t - \tau) e^{j2\pi \nu t}$ follows $\mathcal{CN}(0, \sigma^2)$ with the inflated variance $\sigma^2 = N_0 + |\alpha|^2 \sigma_e^2$.

We define the parameter vector as $\boldsymbol{\theta} = [\tau, \nu, \Re(\alpha), \Im(\alpha)]^T$. The corresponding Fisher Information Matrix (FIM) is:
\begin{equation}
\mathbf{I}(\boldsymbol{\theta}) = \mathbb{E} \left[
\left( \frac{\partial \log p(y|\boldsymbol{\theta})}{\partial \boldsymbol{\theta}} \right)
\left( \frac{\partial \log p(y|\boldsymbol{\theta})}{\partial \boldsymbol{\theta}} \right)^T
\right].
\end{equation}
\textcolor{black}{
The corresponding Cramér--Rao Bound is obtained from
\begin{equation}
\mathrm{CRB}(\boldsymbol{\theta})
=
\mathbf{I}^{-1}(\boldsymbol{\theta}),
\end{equation}
where both the diagonal and off-diagonal entries of
$\mathbf{I}(\boldsymbol{\theta})$ contribute to the estimation
accuracy through parameter coupling.
}
\subsection{FIM Evaluation and Dominant Structure}
For the complex Gaussian model $y(t) \sim \mathcal{CN}(\mu(t;\boldsymbol{\theta}), \sigma^2)$ with mean $\mu(t;\boldsymbol{\theta}) = \alpha \hat{x}(t-\tau) e^{j2\pi \nu t}$,  \textcolor{black}{
each entry of the Fisher Information Matrix can be written as
}
\begin{equation}
\textcolor{black}{
I_{ij}
=
\frac{2}{\sigma^2}
\Re
\left\{
\left(
\frac{\partial\mu}{\partial\theta_i}
\right)^H
\left(
\frac{\partial\mu}{\partial\theta_j}
\right)
\right\}
+
I_{ij}^{(\mathrm{var})},
}
\end{equation}

\textcolor{black}{
where $I_{ij}^{(\mathrm{var})}$ denotes the covariance-gradient
contribution. Under the moderate decoding uncertainty assumption
($\sigma_e^2\ll P$), the variance-gradient terms are negligible, yielding
}
\begin{equation}
\textcolor{black}{
I_{ij}
\approx
\frac{1}
{N_0+|\alpha|^2\sigma_e^2(R,n)}
I_{ij}^{(0)},
}
\end{equation}
where
\[
I_{ij}^{(0)}
=
2
\Re
\left\{
\left(
\frac{\partial\mu}{\partial\theta_i}
\right)^H
\left(
\frac{\partial\mu}{\partial\theta_j}
\right)
\right\}
\]

\textcolor{black}{is the corresponding Fisher Information entry obtained with a perfectly
known reference waveform.}

The partial derivatives of the mean signal with respect to each parameter are:
\begin{align}
\frac{\partial \mu(t)}{\partial \tau} &= -\alpha \dot{\hat{x}}(t-\tau) e^{j2\pi \nu t}, \quad
&\frac{\partial \mu(t)}{\partial \Re(\alpha)} &= \hat{x}(t-\tau) e^{j2\pi \nu t}, \\
\frac{\partial \mu(t)}{\partial \nu} &= j2\pi t \alpha \hat{x}(t-\tau) e^{j2\pi \nu t}, \quad
&\frac{\partial \mu(t)}{\partial \Im(\alpha)} &= j \hat{x}(t-\tau) e^{j2\pi \nu t}.
\end{align}

\textcolor{black}{
Collecting all pairwise Fisher Information terms yields the complete
Fisher Information Matrix
}
\begin{equation}
\textcolor{black}{
\mathbf{I}(\boldsymbol{\theta})=
\begin{bmatrix}
I_{\tau\tau} & I_{\tau\nu} & I_{\tau\Re(\alpha)} & I_{\tau\Im(\alpha)}\\
I_{\nu\tau} & I_{\nu\nu} & I_{\nu\Re(\alpha)} & I_{\nu\Im(\alpha)}\\
I_{\Re(\alpha)\tau} & I_{\Re(\alpha)\nu} &
I_{\Re(\alpha)\Re(\alpha)} &
I_{\Re(\alpha)\Im(\alpha)}\\
I_{\Im(\alpha)\tau} & I_{\Im(\alpha)\nu} &
I_{\Im(\alpha)\Re(\alpha)} &
I_{\Im(\alpha)\Im(\alpha)}
\end{bmatrix},
}
\end{equation}

\textcolor{black}{
where the off-diagonal entries quantify the coupling between the jointly
estimated parameters. The diagonal entries are listed below for
completeness.
}
\begin{align}
I_{\tau\tau} &= \frac{2|\alpha|^2}{\sigma^2} \int |\dot{\hat{x}}(t)|^2 dt, \\
I_{\nu\nu} &= \frac{8\pi^2 |\alpha|^2}{\sigma^2} \int t^2 |\hat{x}(t)|^2 dt, \\
I_{\alpha\alpha} &= \frac{2}{\sigma^2} \int |\hat{x}(t)|^2 dt.
\end{align}

\subsection{Effect of Decoding Uncertainty on CRB}
The Cramér–Rao Bound (CRB) matrix is the inverse of the FIM, $\mathrm{CRB}(\boldsymbol{\theta}) = \mathbf{I}^{-1}(\boldsymbol{\theta})$. Decoding uncertainty degrades the FIM through two mechanisms: \textit{noise amplification} ($\sigma^2 = N_0 + |\alpha|^2 \sigma_e^2$) and \textit{signal distortion} ($\mathbb{E}[|\hat{x}(t)|^2] = P + \sigma_e^2$). 

Consequently, the CRBs formulated for delay and Doppler estimation scale as:
\begin{align}
\mathrm{CRB}(\tau) \propto \frac{\sigma^2}{|\alpha|^2 \int |\dot{\hat{x}}(t)|^2 dt}, \qquad \mathrm{CRB}(\nu) \propto \frac{\sigma^2}{|\alpha|^2 \int t^2 |\hat{x}(t)|^2 dt}.
\end{align}

\textbf{Remark:} \textcolor{black}{
Since every entry of the Fisher Information Matrix is scaled by the same
reliability-dependent covariance inflation factor,
\[
\frac{1}
{N_0+|\alpha|^2\sigma_e^2(R,n)},
\]
both the diagonal and off-diagonal Fisher Information terms inherit the
same finite-blocklength reliability dependence. Consequently, the
Tradeoff Cliff extends naturally to the complete Fisher Information
Matrix and therefore to the joint CRBs for delay, Doppler, and complex
channel gain.
}

%%%%%%%%%%%%%%%%%%%%%%%%%%%%%%%%%%%%%%%%%%%%%%%%%%%%%%%%%%%%%%%%%%%%%%%%%%%%%%%%%%%%%%%%%%%%%%%%%%%%%%%%%%%%%%%%%%%%%%%%%%%%%%%%%%%%%%%%%%%%%%%%%%%%%%%%%%%%%%%%%%%%%%%%%%%%%%%%%%%%%%%%%%%%%%%%%%%%%%%%%%%%%%%%%%%%%%%%%%%%%%%%%%%%%%%%%%%%%%%%%%%%%%%%%%
\bibliographystyle{IEEEtran}
%\begin{thebibliography}{99}
\bibliography{Ref1}

@article{Liu2022_ISAC,
  author = {Liu, F. and Cui, Y. and Masouros, C. and Xu, J. and Han, T. X. and Eldar, Y. C. and Buzzi, S.},
  title = {Integrated Sensing and Communications: Toward Dual-Functional Wireless Networks for {6G} and Beyond},
  journal = {IEEE Journal on Selected Areas in Communications},
  volume = {40},
  number = {6},
  pages = {1728--1767},
  year = {2022},
  doi = {10.1109/JSAC.2022.3156632}
}

@inproceedings{Khan2026_CommSenseLimits,
  author    = {Khan, Mohammed Zafar Ali},
  title     = {Poster: Fundamental Limits of Network Sensing Using {CommSense}},
  booktitle = {Proceedings of IEEE LANMAN}, 
  year      = {2026},
  pages     = {1--3},
  month     = {June},
  dates     ={15,16}
}

@article{Liu2020_ISACFramework,
  author = {Liu, Fan and Masouros, Christos and Petropulu, Athina P. and Griffiths, Hugh and Hanzo, Lajos},
  title = {Joint Radar and Communication Design: Applications, State-of-the-Art, and the Road Ahead},
  journal = {IEEE Transactions on Communications},
  volume = {68},
  number = {6},
  pages = {3834--3862},
  year = {2020}
}

@article{Polyanskiy2010_FBL,
  author = {Polyanskiy, Y. and Poor, H. V. and Verd{\'u}, S.},
  title = {Channel Coding Rate in the Finite Blocklength Regime},
  journal = {IEEE Transactions on Information Theory},
  volume = {56},
  number = {5},
  pages = {2307--2359},
  year = {2010},
  doi = {10.1109/TIT.2010.2043745}
}

@article{Ding2021_IT_ISAC,
  author = {Ding, Z. and Schober, R. and Poor, H. V.},
  title = {On the Limits of Integrated Sensing and Communication Systems},
  journal = {IEEE Transactions on Information Theory},
  volume = {67},
  number = {12},
  pages = {8079--8095},
  year = {2021}
}

@article{Bica2022_Tradeoff,
  author = {Bica, M. and Saligrama, V.},
  title = {A Rate-Distortion Perspective on Integrated Sensing and Communication},
  journal = {IEEE Transactions on Information Theory},
  volume = {68},
  number = {7},
  pages = {4569--4590},
  year = {2022}
}

@INPROCEEDINGS{Arslan2025_PA,
  author={Akca, Huseyin and Memişoğlu, Ebubekir and Çırpan, Hakan Ali and Arslan, Huseyin},
  booktitle={2024 6th International Conference on Communications, Signal Processing, and their Applications (ICCSPA)}, 
  title={Integrated Sensing and Communication with Power Amplifier Impairment}, 
  year={2024},
  volume={},
  number={},
  pages={1-6},
  doi={10.1109/ICCSPA61559.2024.10794378}}

@article{Le2024Hardware,
  author={Zhang, Xue and Le, Ngoc Phuc and Alouini, Mohamed-Slim},
  journal={IEEE Open Journal of Vehicular Technology}, 
  title={{RIS}-Based {DOA} Estimation for Communication-Assisted Sensing Systems Under Hardware Impairments}, 
  year={2025},
  volume={6},
  number={},
  pages={1736-1748},
  doi={10.1109/OJVT.2025.3580041}}

@inproceedings{Ozturk2022,
  author={{\"O}zt{\"u}rk, C. and Wymeersch, H. and others},
  booktitle={IEEE International Conference on Communications (ICC)}, 
  title={On the Impact of Hardware Impairments on {RIS}-aided Localization}, 
  year={2022},
  pages={1-6}
}

@article{Liu2018_MUMIMO,
  author = {Liu, F. and Masouros, C. and Li, A. and Sun, H. and Hanzo, L.},
  title = {{MU-MIMO} Communications With {MIMO} Radar: From Co-Existence to Joint Transmission},
  journal = {IEEE Transactions on Wireless Communications},
  volume = {17},
  number = {4},
  pages = {2755--2770},
  year = {2018}
}

@article{Hassanien2016_RadarComm,
  author = {Hassanien, A. and Amin, M. G.},
  title = {Dual-Function Radar-Communications Using Phase-Modulated Waveforms},
  journal = {IEEE Transactions on Signal Processing},
  volume = {64},
  number = {8},
  pages = {2168--2181},
  year = {2016}
}

@article{Zheng2019_Coexistence,
  author = {Zheng, L. and Lops, M. and Eldar, Y. C. and Wang, X.},
  title = {Radar and Communication Coexistence: An Overview},
  journal = {IEEE Signal Processing Magazine},
  volume = {36},
  number = {5},
  pages = {85--114},
  year = {2019}
}

@inproceedings{Commsense1,
  author={Sardar, Santu and Mishra, Amit K. and Khan, Mohammed Zafar Ali},
  booktitle={2017 IEEE AFRICON}, 
  title={{LTE-CommSense} system and its feasibility analysis}, 
  year={2017},
  pages={1564-1568}
}

@article{Commsense2,
  author={Sardar, Santu and Mishra, Amit K. and Khan, M. Z. A.},
  journal={IEEE Aerospace and Electronic Systems Magazine}, 
  title={{LTE} commsense for object detection in indoor environments}, 
  year={2018},
  volume={33},
  number={7},
  pages={46-59}
}

@article{Zhang2022_DeepWiFi,
  author={Ahmad, Iftikhar and Ullah, Arif and Choi, Wooyeol},
  journal={IEEE Open Journal of the Communications Society}, 
  title={{WiiFi}-Based Human Sensing With Deep Learning: Recent Advances, Challenges, and Opportunities}, 
  year={2024},
  volume={5},
  number={},
  pages={3595-3623},
  doi={10.1109/OJCOMS.2024.3411529}}

@article{Ma2019_WiFi,
  author={Ma, Y. and Zhou, G. and Wang, S. and Zhao, H. and Jung, W.},
  journal={IEEE Communications Surveys \& Tutorials},
  title={{WiFi} Sensing with Channel State Information: A Survey},
  year={2019},
  volume={21},
  number={2},
  pages={1747-1773}
}

@article{Liu2021_Resource,
  author = {Liu, F. and Masouros, C.},
  title = {Toward Dual-Functional Radar-Communication Systems: Optimal Waveform Design},
  journal = {IEEE Transactions on Signal Processing},
  volume = {69},
  pages = {5598--5613},
  year = {2021},
  doi = {10.1109/TSP.2021.3106110}
}

@article{Li2021Mismatched,
  author = {Liu, Rang and Li, Ming and Liu, Qian and Swindlehurst, A. Lee},
  title = {Dual-Functional Radar-Communication Waveform Design: A Symbol-Level Precoding Approach},
  journal = {IEEE Journal of Selected Topics in Signal Processing},
  volume = {15},
  number = {6},
  pages = {1316-1331},
  year = {2021},
  doi = {10.1109/JSTSP.2021.3111438}
}

@article{Zhang2026,
  author = {Zhang, Z. and Masouros, C. and others},
  title = {Fundamental Tradeoffs for {ISAC} Multiple Access in Finite-Blocklength Regime},
  journal = {arXiv preprint arXiv:2601.05165},
  year = {2026}
}

@article{Dong2024CAS,
  author={Dong, Fuwang and Liu, Fan and Lu, Shihang and Xiong, Yifeng and Zhang, Qixun and Feng, Zhiyong and Gao, Feifei},
  journal={IEEE Journal on Selected Areas in Communications}, 
  title={Communication-Assisted Sensing in {{6G}} Networks}, 
  year={2025},
  volume={43},
  number={4},
  pages={1371-1386},
  doi={10.1109/JSAC.2025.3531548}}

@article{Masouros2020,
  author = {Masouros, C. and Liu, F.},
  title = {Dual-functional Radar-Communication Waveform Design in the Presence of Finite-Alphabet Signaling},
  journal = {IEEE Communications Letters},
  volume = {24},
  number = {7},
  pages = {1414-1418},
  year = {2020}
}

@article{Blandino2023,
  author = {Blandino, S. and others},
  title = {Analysis of {OFDM}-Based {{ISAC}} With Hardware Impairments},
  journal = {IEEE Open Journal of the Communications Society},
  volume = {4},
  pages = {1671-1685},
  year = {2023}
}

@article{Demirhan2022,
  author = {Demirhan, U. and Alkhateeb, A.},
  title = {Joint Radar and Communication Beamforming With Digital Metasurfaces},
  journal = {IEEE Transactions on Wireless Communications},
  volume = {21},
  number = {11},
  pages = {9936-9951},
  year = {2022}
}

@inproceedings{Bahl2000_RADAR,
  author = {Bahl, P. and Padmanabhan, V. N.},
  title = {{RADAR}: an in-building {RF}-based user location and tracking system},
  booktitle = {Proceedings IEEE INFOCOM 2000},
  volume = {2},
  pages = {775-784},
  year = {2000}
}

@article{Tan2021_CommSense,
  author = {Tan, B. and others},
  title = {Integrated Terahertz Communication and Radar Sensing: Feasibility, Opportunities, and Challenges},
  journal = {IEEE Wireless Communications},
  volume = {28},
  number = {1},
  pages = {150-158},
  year = {2021}
}

@article{Fortunati2020_CRB_MIMO_CRB,
  author = {Fortunati, S. and others},
  journal = {IEEE Transactions on Signal Processing},
  title = {On the Use of Massive Arrays for Joint Communication and Radar Sensing},
  year = {2020},
  volume = {68},
  pages = {5441--5456},
  doi = {10.1109/TSP.2020.3025078}
}

@article{Lu2024OpenChallenges,
  author={Lu, Shun and others},
  journal={IEEE Internet of Things Journal},
  title={Integrated Sensing and Communications: Recent Advances and Ten Open Challenges},
  year={2024},
  volume={11},
  number={11},
  pages={19094--19120},
  doi={10.1109/JIOT.2024.3361173}
}

@article{Arikan2009,
  author={E. Arikan},
  journal={IEEE Transactions on Information Theory},
  title={Channel Polarization: A Method for Constructing Capacity-Achieving Codes for Symmetric Binary-Input Memoryless Channels},
  year={2009},
  volume={55},
  number={7},
  pages={3051-3073}
}

@ARTICLE{CRB1,
	author={Zhu, Qi and Li, Ming and Liu, Rang and Liu, Qian},
	journal={IEEE Transactions on Wireless Communications}, 
	title={Cramér-Rao Bound Optimization for Active RIS-Empowered {ISAC} Systems}, 
	year={2024},
	volume={23},
	number={9},
	pages={11723-11736},
	doi={10.1109/TWC.2024.3384501}}

@article{TalVardy2015,
  author={I. Tal and A. Vardy},
  journal={IEEE Transactions on Information Theory},
  title={List Decoding of Polar Codes},
  year={2015},
  volume={61},
  number={5},
  pages={2213-2226}
}

@ARTICLE{PER1,
	author={Meng, Kaitao and Masouros, Christos and Petropulu, Athina P. and Hanzo, Lajos},
	journal={IEEE Transactions on Wireless Communications}, 
	title={Cooperative {ISAC Networks}: Performance Analysis, Scaling Laws, and Optimization}, 
	year={2025},
	volume={24},
	number={2},
	pages={877-892},
	doi={10.1109/TWC.2024.3491356}}

@ARTICLE{Wav1,
	author={Lee, Byunghyun and Kim, Hwanjin and Love, David J. and Krogmeier, James V.},
	journal={IEEE Transactions on Wireless Communications}, 
	title={Spatial-Division ISAC: A Practical Waveform Design Strategy via Null-Space Superimposition}, 
	year={2026},
	volume={25},
	number={},
	pages={6837-6851},
	doi={10.1109/TWC.2025.3627260}}

@ARTICLE{IPSAC,
	author={Lee, Byunghyun and Liu, Rang and Love, David J. and Krogmeier, James V. and Lee Swindlehurst, A.},
	journal={IEEE Transactions on Wireless Communications}, 
	title={Integrated Polarimetric Sensing and Communication With Polarization-Reconfigurable Arrays}, 
	year={2026},
	volume={25},
	number={},
	pages={10618-10634},
	doi={10.1109/TWC.2026.3653252}}

@ARTICLE{Xiong2023_Fundamental,
  author={Xiong, Yifeng and Liu, Fan and Cui, Yuanhao and Yuan, Weijie and Han, Tony Xiao and Caire, Giuseppe},
  journal={IEEE Transactions on Information Theory}, 
  title={On the Fundamental Tradeoff of Integrated Sensing and Communications Under Gaussian Channels}, 
  year={2023},
  volume={69},
  number={9},
  pages={5723-5751},
  doi={10.1109/TIT.2023.3284449}}

@article{wei2021joint,
  title={Joint communication and radar sensing in {6G} wireless systems: Opportunities and challenges},
  author={Wei, Zhiqiang and Yuan, Weijie and Li, Shuangyang and Yuan, Jinhong and Ng, Derrick Wing Kwan and Xia, Xiang-Gen},
  journal={IEEE Wireless Communications},
  volume={28},
  number={1},
  pages={149--159},
  year={2021},
  publisher={IEEE}
}

@article{mezghani2022blind,
  title={Blind self-sensing for {ISAC} systems with hardware impairments},
  author={Mezghani, Amine and Swindlehurst, A Lee},
  journal={IEEE Transactions on Signal Processing},
  volume={70},
  pages={3193--3208},
  year={2022},
  publisher={IEEE}
}

@article{shehab2023finite,
  title={Finite blocklength performance of {ISAC} for {URLLC}},
  author={Shehab, Mohammad and Abbas, Rana and Almradi, Akram and Al-Dhahir, Naofal and Schober, Robert},
  journal={IEEE Open Journal of the Communications Society},
  year={2023},
  publisher={IEEE}
}

@article{condo2021short,
  title={Short-packet polar codes: A survey on design and performance},
  author={Condo, Carlo},
  journal={IEEE Communications Surveys \& Tutorials},
  volume={23},
  number={4},
  pages={2451--2482},
  year={2021},
  publisher={IEEE}
}

@ARTICLE{hua2023mimo,
  author={Hua, Haocheng and Han, Tony Xiao and Xu, Jie},
  journal={IEEE Transactions on Wireless Communications}, 
  title={MIMO Integrated Sensing and Communication: CRB-Rate Tradeoff}, 
  year={2024},
  volume={23},
  number={4},
  pages={2839-2854},
  doi={10.1109/TWC.2023.3303326}}

@misc{Pucci2025_FI_ISAC,
	title={Position and Velocity Estimation Accuracy in MIMO-OFDM ISAC Networks: A Fisher Information Analysis}, 
	author={Lorenzo Pucci and Luca Arcangeloni and Andrea Giorgetti},
	year={2025},
	eprint={2507.01743},
	archivePrefix={arXiv},
	primaryClass={eess.SP},
	url={https://arxiv.org/abs/2507.01743}, 
}

@ARTICLE{Yifeng,
	author={Liu, An and Huang, Zhe and Li, Min and Wan, Yubo and Li, Wenrui and Han, Tony Xiao and Liu, Chenchen and Du, Rui and Tan, Danny Kai Pin and Lu, Jianmin and Shen, Yuan and Colone, Fabiola and Chetty, Kevin},
	journal={IEEE Communications Surveys \& Tutorials}, 
	title={A Survey on Fundamental Limits of Integrated Sensing and Communication}, 
	year={2022},
	volume={24},
	number={2},
	pages={994-1034},
	doi={10.1109/COMST.2022.3149272}}

@techreport{Vuong1986,
  author      = {Quang H. Vuong},
  title       = {Cram{\'e}r--Rao Bounds for Misspecified Models},
  institution = {Division of the Humanities and Social Sciences, California Institute of Technology},
  type        = {Working Paper},
  number      = {652},
  address     = {Pasadena, CA, USA},
  month       = oct,
  year        = {1986}
}

@article{Richmond2015,
  author  = {C. D. Richmond and L. L. Horowitz},
  title   = {Parameter Bounds on Estimation Accuracy Under Model Misspecification},
  journal = {IEEE Transactions on Signal Processing},
  volume  = {63},
  number  = {9},
  pages   = {2263--2278},
  month   = may,
  year    = {2015},
  doi     = {10.1109/TSP.2015.2407314}
}

@article{Fortunati2016,
  author  = {Stefano Fortunati and Fulvio Gini and Maria S. Greco},
  title   = {The Misspecified Cram{\'e}r--Rao Bound and Its Application to Scatter Matrix Estimation in Complex Elliptically Symmetric Distributions},
  journal = {IEEE Transactions on Signal Processing},
  volume  = {64},
  number  = {9},
  pages   = {2387--2399},
  month   = may,
  year    = {2016},
  doi     = {10.1109/TSP.2016.2522433}
}

@book{Kay1993,
  author    = {Steven M. Kay},
  title     = {Fundamentals of Statistical Signal Processing, Volume I: Estimation Theory},
  publisher = {Prentice Hall},
  year      = {1993}
}

@article{Liu2025RandomISAC,
  author  = {Fan Liu and Ya-Feng Liu and Yuanhao Cui and Christos Masouros and Jie Xu and Tony Xiao Han and Stefano Buzzi and Yonina C. Eldar and Shi Jin},
  title   = {Sensing With Communication Signals: From Information Theory to Signal Processing},
  journal = {IEEE Communications Surveys \& Tutorials},
  year    = {2025},
  note    = {Early Access}
}

@article{Lu2025RandomSignals,
  author  = {Shihang Lu and Fan Liu and Yifeng Xiong and Zhen Du and Yuanhao Cui and Shuangyang Li and Weijie Yuan and Jie Yang and Shi Jin},
  title   = {Sensing With Random Communication Signals},
  journal = {arXiv preprint arXiv:2504.06537},
  year    = {2025}
}

@article{Shen2026TIT,
  author  = {X. Shen and Z. Lu and N. Zhao and H. Zhao and Y. Shen},
  title   = {Fundamental Tradeoff of Bistatic ISAC Under Gaussian Fading Channels at Finite Blocklength},
  journal = {IEEE Transactions on Information Theory},
  volume  = {72},
  number  = {2},
  pages   = {1176--1200},
  month   = feb,
  year    = {2026},
  doi     = {10.1109/TIT.2025.3623189}
}

@article{Lin2026TWC,
  author  = {X. Lin and X. Sun and J. Li and P. Zhu and D. Wang and F. Shu and X. You},
  title   = {Adaptive Finite-Blocklength Optimization for the Communication--Sensing Tradeoff in Network-Assisted Full-Duplex Cell-Free ISAC Systems With URLLC Users},
  journal = {IEEE Transactions on Wireless Communications},
  volume  = {25},
  pages   = {17446--17462},
  year    = {2026},
  doi     = {10.1109/TWC.2026.3693349}
}

@book{HornJohnson2013,
  author    = {Roger A. Horn and Charles R. Johnson},
  title     = {Matrix Analysis},
  edition   = {2},
  publisher = {Cambridge University Press},
  address   = {Cambridge, UK},
  year      = {2013},
  isbn      = {9780521548236}
}

@ARTICLE{dfd1,
  author={Bere, Praveen Sai and Khan, Mohammed Zafar Ali and Hanzo, Lajos},
  journal={IEEE Open Journal of Vehicular Technology}, 
  title={A Low-Complexity Diversity-Preserving Universal Bit-Flipping Enhanced Hard Decision Decoder for Arbitrary Linear Codes}, 
  year={2024},
  volume={5},
  number={},
  pages={1496-1517},
  doi={10.1109/OJVT.2024.3437470}}

@article{dfd2,
author = {Sai, Praveen and Khan, Mohammed and Hanzo, L.},
year = {2026},
month = {05},
pages = {1-17},
title = {A Universal Generalized Flip Decoder for Next-Generation URLLC Systems},
journal = {IEEE Transactions on Vehicular Technology},
doi = {10.1109/TVT.2026.3697687}
}
%\end{thebibliography}

\end{document}